\documentclass[letterpaper,12pt]{article}

\usepackage[margin=0.9in,letterpaper]{geometry}

\usepackage[english]{babel}
\usepackage[utf8]{inputenc}
\usepackage{amsmath,amssymb}
\usepackage{scalerel}
\usepackage{parskip}
\usepackage{graphicx}
\usepackage{mathrsfs}  
\usepackage{mathtools}
\usepackage{cancel}
\usepackage{tcolorbox}
\usepackage{amsmath, bm}
\usepackage{babel}
\usepackage{csquotes}
\usepackage{hyperref}
\usepackage{amsthm}

\DeclareUnicodeCharacter{0327}{}

\usepackage[style=ieee,backend=biber]{biblatex}

\usepackage{tikz-cd}
\usetikzlibrary{cd}

\usepackage[shortlabels]{enumitem}

\usepackage{xcolor}

\makeatother

\newtheorem{theorem}{Theorem}

\newtheorem{lemma}{Lemma}

\newtheorem{definition}{Definition}

\newtheorem{remark}{Remark}

\usepackage{mathrsfs}  

\usepackage{xparse}

\title{Classical dynamical $r$-matrices for the Chern-Simons formulation of generalised 3d gravity}
\author{Juan Carlos Morales Parra{$\dagger$} \quad Bernd J Schroers{$\ddagger$} \\ $\dagger$ \small{Maxwell Institute of Mathematical Sciences and Department of Mathematics,} \\ \small{Heriot-Watt University, Edinburgh EH14 4AS, United Kingdom} \\ \small{\url{jcm2000@hw.ac.uk}} \\
$\ddagger$ \small{Maxwell Institute of Mathematical Sciences and School of Mathematics,} \\ \small{University of Edinburgh, Edinburgh EH9 3FD, United Kingdom} \\ \small{\url{b.schroers@ed.ac.uk}} \\ }
\date{\today}

\begin{document}

\maketitle
\begin{center}
\textbf{Abstract} 
\end{center}

Classical dynamical $r$-matrices arise naturally in the combinatorial description of the phase space of Chern-Simons theories, either through the inclusion of dynamical sources or through a gauge-fixing procedure involving two punctures. Here we consider classical dynamical $r$-matrices for the family of Lie algebras which arise in the Chern-Simons formulation of 3d gravity, for any value of the cosmological constant. We derive differential equations for classical dynamical $r$-matrices in this case, and show that they can be viewed as  generalised complexifications, in a sense which we define, of the equations governing dynamical $r$-matrices for $\mathfrak{su}(2)$ and $\mathfrak{sl}(2,\mathbb{R})$. We obtain explicit families of  solutions and relate them, via Weierstrass factorisation,  to solutions found by Feher, Gabor, Marshall, Palla and Pusztai in the context of chiral WZWN models.

\section{Introduction}\label{Sec1paper1}


\subsection{Motivation}

One important reason for studying classical dynamical $r$-matrices, and the one which motivates this paper,  is their role in the gauge-fixed Poisson structures of character varieties over punctured Riemann surfaces. Such  Poisson spaces  appear in particular as phase-spaces of lower dimensional gauge theories, like the Chern-Simons formulation of gravity in three dimensions (see e.g. \cite{schroers2001combinatorial}, \cite{Buffenoir:2005zi}, \cite{Meusburger_2011}, \cite{Meusburger:2011oew}, \cite{Meusburger:2012wc}). As we shall explain, the gauge fixing is interesting and natural classically. 
In addition, it offers a route to quantisation  of constraint systems which avoids some of the technical challenges of imposing constraints after quantisation. This is of particular interest in the Chern-Simons formulation of 3d gravity where the quantised theory is an interesting toy model for quantum gravity.

The goal of this paper is  to study the space of classical dynamical $r$-matrices  up to a naturally defined notion of gauge equivalence for some particular Lie algebras $\mathfrak{g}_\lambda$ that appear in the setting of 3d gravity,  namely the Lie algebras of the symmetry groups of maximally symmetric three dimensional Riemannian and pseudo-Riemannian manifolds. In this extended introduction we define classical dynamical $r$-matrices, illustrate their appearance in  gauge-fixed character varieties in a simple example and then prepare the ground for a full classification of  the  classical dynamical $r$-matrices which are relevant for 3d gravity in the remainder of this paper.

\subsection{Definitions}

A classical dynamical $r$-matrix is an equivariant solution of the Classical Dynamical Yang-Baxter Equation. More precisely, given a finite dimensional Lie algebra $\mathfrak{g}$, a Lie subalgebra $\mathfrak{h} \subseteq \mathfrak{g}$ and an element $K \in (S^2 \mathfrak{g})^{\mathfrak{g}}$, a classical dynamical $r$-matrix associated to the triple $(\mathfrak{g},\mathfrak{h},K)$ is a $\mathfrak{h}$-equivariant function $K + r \in C^1(\mathfrak{h}^*) \otimes (\mathfrak{g} \otimes \mathfrak{g})$ that solves the Classical Dynamical Yang-Baxter Equation (CDYBE)

\begin{equation}\label{CDYB}
\text{CDYB}(r) \equiv [[r,r]]+\text{Alt}(d r)=-[[K,K]].
\end{equation}

Explicitly, 
\begin{equation*}
[[r,r]]= \text{CYB}(r) \equiv [r_{12},r_{23}]+[r_{13},r_{23}]+[r_{12},r_{13}]
\end{equation*}

is the standard Schouten bracket and 
\begin{equation*}
\text{Alt}(dr) \equiv  
  h_i^{(1)} \frac{\partial r_{23}}{\partial h^i} - h_i^{(2)} \frac{\partial r_{13}}{\partial h^i} + h_i^{(3)} \frac{\partial r_{12}}{\partial h^i}, 
\end{equation*}
in terms of a basis $\{h_i\}_{i=1,\cdots, \dim \mathfrak{h}}$ of $\mathfrak{h}$ and its dual $\{h^i\}_{i=1,\cdots, \dim \mathfrak{h}}$ for $\mathfrak{h}^*$, where  

\begin{equation*}
    h_i^{(1)}\frac{\partial r_{23}}{\partial h^i}  =  \frac{\partial r^{ab}}{\partial h^i} h_i \otimes T_a \otimes T_b, \quad 
    h_i^{(2)}\frac{\partial r_{13}}{\partial h^i}  =  \frac{\partial r^{ab}}{\partial h^i} T_a \otimes h_i \otimes T_b , \quad h_i^{(3)}\frac{\partial r_{12}}{\partial h^i}  = \frac{\partial r^{ab}}{\partial h^i}  T_a \otimes T_b \otimes h_i,
\end{equation*}
with $r=\sum_{ab}r^{ab}T_a \otimes T_b$ and $\{T_a\}_{a=1,\cdots,\dim \mathfrak{g}}$ a basis of $\mathfrak{g}$. The Einstein summation convention will be used throughout the paper.\\\\
The $\mathfrak{h}$-equivariance condition means 
\begin{equation}\label{equiv}
\frac{d }{ds}\bigg|_{s=0} r( \text{Ad}^*(e^{sh}) x)  + [r(x),h \otimes 1 + 1 \otimes h] = 0
\end{equation}
holds for all $x \in \mathfrak{h}^*$ and $h \in \mathfrak{h}$, i.e. the coadjoint action of $\mathfrak{h}$ on the argument of $r$ equals to the adjoint action of $\mathfrak{h}$ on the Lie algebra part of $r$. \\\\
The CDYBE originally appeared in research  related to integrability in  conformal field theories (see e.g. \cite{Felder:1994pb}, \cite{BALOG1990227}). Since then, this equation, its solutions  and its quantum counterpart have been widely studied, especially its applications to the theory of integrable systems, (Quasi-) Poisson geometry and special functions (see e.g. \cite{etingof_schiffmann_2002}, \cite{Etingof2002OnTD}, \cite{alma9931324753902959}).\\\\
The definition of a classical dynamical $r$-matrix varies according to the reference under consideration,  and may include or omit the equivariance condition. For example, in literature on the origin and application of solutions of the CDYBE to (Quasi-)Poisson structures,  equivariance is part of the definition (see e.g. \cite{AIF_2001__51_3_835_0}), while in literature concerning their application to quantum integrable systems the equivariance condition is usually not considered (see e.g. \cite{sechin2021quadratic}). Here we make the distinction between solutions to the CDYBE and classical dynamical $r$-matrices, where the later are obtained from the former after imposing the equivariance condition (\ref{equiv}). \\\\
Let $G$ be a Lie group and denote by $\mathfrak{g}$ its Lie algebra. Given a Lie subalgebra $\mathfrak{h}$ and an element $K \in (S^2 \mathfrak{g})^{\mathfrak{g}}$, the set of classical dynamical ($\mathfrak{g},\mathfrak{h},K$) $r$-matrices is denoted by $\text{Dyn}(\mathfrak{g},\mathfrak{h},K)$. By taking the union of these spaces for all Lie subalgebras $\mathfrak{h}$ of $\mathfrak{g}$, we get the set of dynamical $r$-matrices associated to the pair $(\mathfrak{g},K)$, denoted by 
\begin{equation}\label{Dyn}
  \text{Dyn}(\mathfrak{g},K)= \bigcup_{\mathfrak{h} \leq \mathfrak{g}} \text{Dyn}(\mathfrak{g},\mathfrak{h},K),
\end{equation}
whose geometry and structure are studied e.g. in \cite{etingof1998geometry} and \cite{Etingof2000OnTM}. \\\\
If $\mathfrak{h}$ and $\mathfrak{h}'$ are conjugate-equivalent Lie subalgebras of $\mathfrak{g}$, say by $g \in G$, then the map
\begin{equation}\label{mapdyn}
\begin{split}
    \text{Dyn}(\mathfrak{g},\mathfrak{h},K)  & \to \text{Dyn}(\mathfrak{g},\mathfrak{h}',K)  \\
    r(x) & \mapsto \text{Ad}(g) \otimes \text{Ad}(g) r(\text{Ad}^*(g^{-1})x)
    \end{split}
\end{equation}
is well-defined and bijective (see Lemma A in Appendix A), allowing us to define an action of the Lie group $G$ over the space (\ref{Dyn}). 

In this paper, as is usual in the literature regarding the CDYBE, we focus in the case where the Lie subalgebras considered are Cartan subalgebras of $\mathfrak{g}$. Along the rest of the paper $\mathfrak{h}$ will be mainly used to name a Cartan subalgebra of a Lie algebra. \\\\ 
Let $G$ be a Lie group and denote by $\mathfrak{C}_\mathfrak{g}$ the set of Cartan subalgebras of the Lie algebra $\mathfrak{g}$ and by $\mathfrak{C}^{\text{Ad}}_\mathfrak{g}$ the set of conjugacy classes of Cartan subalgebras of $\mathfrak{g}$. Analogous to the construction of the set of classical dynamical $r$-matrices $\text{Dyn}(\mathfrak{g},K)$, we restrict to Cartan subalgebras and define the set of \textit{Cartan} classical dynamical $r$-matrices associated to $(\mathfrak{g},K)$ by 
\begin{equation}\label{dynmodcart}
\text{Dyn}^{\mathfrak{C}}(\mathfrak{g},K)=\bigcup_{\mathfrak{h}\in \mathfrak{C}_\mathfrak{g}} \text{Dyn}(\mathfrak{g},\mathfrak{h},K).
\end{equation}
By picking one representative in each conjugacy class of $\mathfrak{C}_\mathfrak{g}^{\text{Ad}}$ and collecting them in a set denoted by $[\mathfrak{C}^{\text{Ad}}_\mathfrak{g}]$, the set (\ref{dynmodcart}) cab be decomposed as  
\begin{equation*}
\text{Dyn}^{\mathfrak{C}}(\mathfrak{g},K) = \bigcup_{g \in G} \bigsqcup_{\mathfrak{h} \in [\mathfrak{C}_\mathfrak{g}^{\text{Ad}}]}  \text{Dyn}^g(\mathfrak{g},\mathfrak{h},K),
    \end{equation*}
where $\text{Dyn}^g(\mathfrak{g},\mathfrak{h},K)$ is a short notation for the image under the map (\ref{mapdyn}) for $g \in G$. \\\\  
Over the set of classical dynamical $r$-matrices $\text{Dyn}(\mathfrak{g},\mathfrak{h},K)$, for each $\mathfrak{h} \in [\mathfrak{C}_\mathfrak{g}^{\text{Ad}}]$, the group of smooth functions from $\mathfrak{h}^*$ to the stabilizer of $H$
\begin{equation}\label{dyngaugetrans}
    \mathcal{G}(\mathfrak{g},\mathfrak{h}) \equiv \{p \in \text{Fun}(\mathfrak{h}^*, G^H) \ | \ p \text{ is smooth} \},
\end{equation} 
the so-called \textit{group of dynamical $(\mathfrak{g},\mathfrak{h})$ gauge transformations}, where $H \leq G$ is the Lie subgroup of $G$ such that $\text{Lie}(H)=\mathfrak{h}$, defines an action given by 
\begin{equation}\label{dynact}
    p \rhd r = \text{Ad}(p) \otimes \text{Ad}(p)  \left( r + \overline{\eta}^p - \overline{\eta}^p_{21}  \right) \qquad \text{ for } r \in \text{Dyn}(\mathfrak{g},\mathfrak{h},K),
\end{equation}
where $\overline{\eta}^p:\mathfrak{h}^* \to \mathfrak{h} \otimes \mathfrak{g}^\mathfrak{h}$ is the dual of the $\mathfrak{g}^\mathfrak{h}$-valued 1-form $\eta^p=p^{-1}dp$, i.e. explicitly 
\begin{equation*}
    \overline{\eta}^p = \sum_{i=1}^{\dim \mathfrak{h}} h_i \otimes p^{-1} \frac{\partial p}{\partial h^{i}}
\end{equation*}
in terms of the dual bases $\{h_i\}_{i=1,\cdots,\dim \mathfrak{h}}$ and $\{ h^i\}_{i=1,\cdots,\dim \mathfrak{h}}$ for $\mathfrak{h}$ and $\mathfrak{h}^*$, respectively. \\\\
The moduli space of \textit{Cartan} classical dynamical $r$-matrices associated to $(\mathfrak{g},K)$ is defined by 
\begin{equation}\label{decomp}
\mathcal{M}^{\mathfrak{C}}(\mathfrak{g},K)= \bigsqcup_{\mathfrak{h} \in [\mathfrak{C}_\mathfrak{g}^{\text{Ad}}]} \mathcal{M}(\mathfrak{g},\mathfrak{h},K) \end{equation} 
where 
\begin{equation*}
\mathcal{M}(\mathfrak{g},\mathfrak{h},K) \equiv \text{Dyn} (\mathfrak{g},\mathfrak{h},K)/\mathcal{G}(\mathfrak{g},\mathfrak{h}), 
\end{equation*}
such that the full set of \text{Cartan} classical dynamical $r$-matrices (\ref{dynmodcart}) can be generated from it via (\textit{i}) dynamical $(\mathfrak{g},\mathfrak{h})$ gauge transformations (\ref{dynact}) over $\mathcal{M}(\mathfrak{g},\mathfrak{h},K)$ for each $\mathfrak{h} \in [\mathfrak{C}_g^{\text{Ad}}]$ and (\textit{ii}) the $G$-action (\ref{mapdyn}) over $\text{Dyn}(\mathfrak{g},\mathfrak{h},K)$ for each $g \in G$ and $\mathfrak{h} \in [\mathfrak{C}_g^{\text{Ad}}]$.

As will be explained below, the Poisson structures of gauge-fixed moduli spaces of $G$-flat connections (character varieties) over Riemann surfaces are in bijection with the moduli space of Cartan dynamical $r$-matrices $\mathcal{M}^\mathfrak{C}(\mathfrak{g},K)$. The decomposition in (\ref{decomp}) implies that in order to have a full-description of these Poisson structures, it is sufficient  to determine classical dynamical $r$-matrices associated to $(\mathfrak{g},\mathfrak{h},K)$ for representatives $\mathfrak{h}$ of each conjugacy class in $\mathfrak{C}^{\text{Ad}}_\mathfrak{g}$, up to dynamical $(\mathfrak{g},\mathfrak{h})$ gauge transformations. 
\subsection{Example}

 To illustrate the appearance of dynamical $r$-matrices in a simple setting, consider the character variety 

\begin{equation*}
\mathcal{P}^{0,4}_{\text{SL}(2,\mathbb{R}),\{\mathcal{C}_i\}} \equiv \{ A \in \text{Hom}(\pi_1(\Sigma_{0,4}),\text{SL}(2,\mathbb{R}))|A(\ell_i) \in \mathcal{C}_i \text{ for } i=1,2,3,4\} / \text{SL}(2,\mathbb{R})
\end{equation*}

where $\Sigma_{0,4} \equiv \mathbb{S}^2 - \{4 \text{pts}\}$ is the four punctured 2-sphere, $\{\ell_i\}_{i=1,\cdots,4}$ is the (homotopy type) set of generators of its fundamental group consisting of four loops (each going around one of the punctures once)  and the $\mathcal{C}_i$'s are fixed conjugacy classes of $\text{SL}(2,\mathbb{R})$. For simplicity we assume the conjugacy classes are two dimensional (over $\mathbb{R}$) generated by elements of the form $e^{s_iJ_0}$ for $i=1,2,3,4$, where $\{J_a\}_{a=0,1,2}$ is the standard basis of the real Lie algebra $\mathfrak{sl}(2,\mathbb{R})$ satisfying 
\begin{equation*}
[J_a,J_b]=\epsilon_{abc}J^c,
\end{equation*}
and the coefficients are raised or lowered using the Minkowskian metric $\text{diag}(1,-1,-1)$. \\\\
The canonical Poisson structure over $\mathcal{P}^{0,4}_{\text{SL}(2,\mathbb{R}),\{\mathcal{C}_i\}}$ can be seen as a reduction of the Poisson structure over the extended space 
\begin{equation*}
\mathcal{P}^{0,4}_{\text{SL}(2,\mathbb{R}),\text{ext}} \equiv \underbrace{\text{SL}(2,\mathbb{R})_{1} \times \cdots \times \text{SL}(2,\mathbb{R})_4}_{\text{4 times}},  
\end{equation*}
with Poisson bivector (see \cite{Fock:1998nu}) given by 
\begin{equation}\label{FR04}
\Pi^{0,4}_{\text{ext}}(r)  \equiv r^{ab}  \sum_{1 \leq i \leq j \leq 4}  (R_a^{i} + L_a^{i}) \wedge (R_b^{j} + L_b^{j}), 
\end{equation}
where $r=r^{ab}J_a \otimes J_b \in \mathfrak{sl}(2,\mathbb{R}) \otimes \mathfrak{sl}(2,\mathbb{R})$ is a classical $r$-matrix, i.e. a solution of the Classical Yang-Baxter Equation (CYBE) 
\begin{equation}
[[r,r]]=0,
\end{equation} 
such that its symmetric part $K \equiv r+r^{21}$ is $\mathfrak{sl}(2,\mathbb{R})$-invariant. \\\\
In the expression above we have adopted the notation of $R_a^i$ and $L_a^i$ to indicate the right and left fundamental vector fields generated by the Lie algebra element $J_a$ and associated to the $i$-th copy of $\text{SL}(2,\mathbb{R})$, respectively. \\\\
The Poisson space $\mathcal{P}^{0,4}_{\text{SL}(2,\mathbb{R}),\{\mathcal{C}_i\}}$ is obtained from $(\mathcal{P}_{\text{SL}(2,\mathbb{R}),\text{ext}}^{0,4},\Pi_{\text{ext}}^{0,4})$ by ($i$) restricting each copy of $\text{SL}(2,\mathbb{R})$ to the corresponding conjugacy class, ($ii$) imposing the \textit{topological} condition that the product of the elements in each $4$-tuple must be the identity element of $\text{SL}(2,\mathbb{R})$, and ($iii$) identifying conjugate-equivalent $4$-tuples. This approach, developed originally by Fock and Rosly (see e.g. \cite{Fock:1993ms}, \cite{Fock:1998nu}, \cite{Mouquin2016TheFP}), exhibits explicitly how the character variety $\mathcal{P}^{0,4}$ can be realized as a \textit{constrained system}.     \\\\
The dimension of $\mathcal{P}_{\text{SL}(2,\mathbb{R}),\text{ext}}^{0,4}$ is $3 \times 4=12$, while the dimension of $\mathcal{P}^{0,4}_{\text{SL}(2,\mathbb{R}),\{\mathcal{C}_i\}}$ is 
\begin{equation*}
2=12-10=12-\underbrace{4 \times (1)}_{(i)}-\underbrace{3}_{(ii)}-\underbrace{3}_{(iii)}
\end{equation*}
where the underbrace symbols indicate which of the three types of constraints indicated above is responsible of the corresponding dimensional reduction. Since the constraint functions over $\mathcal{P}^{0,4}_{\text{ext}}$  associated to ($i$) are Poisson functions with respect to $\Pi_{\text{ext}}^{0,4}$ (see \cite{meusburger2003poisson}), as an intermediate step we have the partially-constrained space given by            
\begin{equation*}
    \mathcal{P}_{\text{SL}(2,\mathbb{R}),\text{ext(cc)}}^{0,4} \equiv \mathcal{C}_1 \times \mathcal{C}_2 \times \mathcal{C}_3 \times \mathcal{C}_4 
\end{equation*}
of dimension $8=4 \times 2$, with Poisson structure given still by the bivector $\Pi_{\text{ext}}^{0,4}$, in such a way that the fully reduced space is obtained by imposing the remaining $6=3+3$ constraints (topological and conjugation equivalence constraints). The label $(\text{cc})$ is used to indicate each copy of $\text{SL}(2,\mathbb{R})$ has been restricted to the conjugacy class of the corresponding puncture. \\\\
The quantization of character varieties (see e.g. \cite{ben2018quantum}) has been an active topic of research in the recent years, both in mathematics and physics, since it provides (e.g.) quantum group representations of the mapping class group of Riemann surfaces and a quantization scheme for Chern-Simons theory. The fact that character varieties can be realized as constrained Poisson spaces implies one needs to deal with the constraints at some point along the way to quantization. Even thought the standard approach is to impose the constraints at the quantum level (see e.g. \cite{alekseev1995combinatorial}, \cite{alekseev1996combinatorial}), there are technical advantages to  incorporating the constraints at the classical level and then quantizing the reduced theory\footnote{The two approaches are widely assumed to be equivalent, but there is no proof of this assumption.}. \\\\
At the classical level the (first class) constraints of any constrained Poisson space can be \textit{gauge fixed}, which amounts to reducing the space to a sector of the original one with the help of auxiliary (second class) constraints and defining over it a new Poisson structure (the so-called \textit{Dirac brackets}, see \cite{dirac_1950} and \cite{dirac_1951}), in such a way that the constraints become Poisson functions (for details see \cite{9d50b89c-daf8-39a3-a319-7c19d071d0ae}, \cite{FigueroaOFarrill1989BRSTCA}, \cite{matschull1996dirac}). In the particular case of $\mathcal{P}^{0,4}_{\text{SL}(2,\mathbb{R}),\{\mathcal{C}_i\}}$, the three topological constraints ($ii$) can be gauge fixed \textit{\'{a} la} Dirac via three auxiliary constraint functions defined just on the part $\mathcal{C}_1 \times \mathcal{C}_2$ of the space (see \cite{Meusburger:2012wc} for details), getting a $5=1+2+2$ dimensional intermediate space 
\begin{equation*}
\mathcal{P}^{0,4}_{\text{SL}(2,\mathbb{R}),\{\mathcal{C}_i\},GF(1,2)}= H_0 \times \mathcal{C}_3 \times \mathcal{C}_4
\end{equation*}

where $H_0$ is the Cartan group generated by the Lie algebra generator $J_0$, with Poisson structure given explicitly for $F,\tilde{F} \in C^{\infty}(\text{SL}(2,\mathbb{R}) \times \text{SL}(2,\mathbb{R}))$ by
\begin{equation}\label{DFR04}
\begin{split}
\{F, \varphi \} & = \sum_{i=3,4}(R_0^i + L_0^i)F, \\
\{ F,\tilde{F} \}(\varphi) & = \Pi_{\text{ext}}^{0,4}(r(\varphi))(dF,d\tilde{F})= r^{ab}(\varphi) \sum_{3 \leq i \leq j \leq 4} (R^i_a + L_a^i) \wedge (R_{b}^j + L_b^j)(dF,d\tilde{F}),
\end{split}
\end{equation}
where $\varphi$ is a variable parametrizing the dual of the Lie subalgebra $\mathfrak{h}_0=\text{Lie}(H_0)$ and 
\begin{equation*}
r:\mathfrak{h}_0^* \to \mathfrak{sl}(2,\mathbb{R}) \otimes \mathfrak{sl}(2,\mathbb{R})
\end{equation*}
is now an $\mathfrak{h}_0$-equivariant solution of the CDYBE (\ref{CDYB})  such that $\text{Sym}(r(\varphi))=\text{Sym}(r) \equiv K$, i.e. a classical dynamical $r$-matrix for the triple $(\mathfrak{sl}(2,\mathbb{R}),\mathfrak{h}_0,K)$. In fact, the Jacobi identity for the brackets of the form $\{F,\{\tilde{F},\varphi\}\}$ and $\{F,\{\tilde{F},\overline{F}\}\}$ (for $F,\tilde{F},\overline{F} \in C^\infty(\text{SL}(2,\mathbb{R}) \times \text{SL}(2,\mathbb{R}))$) implies the $\mathfrak{h}_0$-equivariance of $r$ and $\text{CDYB}(r)=0$, respectively. \\\\
Hence after performing the gauge fixing, an explicit realization of the Poisson space $\mathcal{P}^{0,4}_{\text{SL}(2,\mathbb{R}),\{\mathcal{C}_i\}}$ is obtained by imposing the now Poisson constraint  
\begin{equation*}
g_4^{-1} g_3^{-1} = e^{ \varphi J_0},
\end{equation*}
obtaining in this way the two dimensional \textit{Atiyah-Bott} phase space of $\text{SL}(2,\mathbb{R})$-flat connections over the four-punctured 2-sphere (see e.g. \cite{Alekseev2000QuasiPoissonM}, \cite{Xu2014GeneralizedCD}). Analogously, if instead conjugacy classes generated by elements of the form $e^{s_iJ_1}$ are considered, the Poisson structure of the gauge-fixed phase space will be defined in terms of a classical dynamical $r$-matrix for the triple $(\mathfrak{sl}(2,\mathbb{R}),\mathfrak{h}_1,K)$, where $\mathfrak{h}_1$ is the Lie subalgebra generated by $J_1$. \\\\
This example helps to understand how the Poisson structures of the moduli space $\mathcal{P}^{0,4}$ are determined by the moduli space of Cartan classical dynamical $r$-matrices $\mathcal{M}^\mathfrak{C}(\mathfrak{sl}(2,\mathbb{R}),K)$. Indeed, due to (\ref{decomp}) and the fact any Cartan subalgebra of $\mathfrak{sl}(2,\mathbb{R})$ is either conjugate to $\mathfrak{h}_0$ or to $\mathfrak{h}_1$, the problem of describing the Poisson structures of $\mathcal{P}^{0,4}$ reduces to finding two classical dynamical $(\mathfrak{sl}(2,\mathbb{R}),\mathfrak{h},K)$ $r$-matrices, one for $\mathfrak{h}=\mathfrak{h}_0$ and for $\mathfrak{h}=\mathfrak{h}_1$.  
\subsection{Chern-Simons formulation of 3d gravity}
Character varieties of the type
\begin{equation*}
\mathcal{P}^{g,n}_{G,\{C_i\}}
 \equiv \{A \in \text{Hom}(\pi_1(\Sigma_{g,n}),G)| A(\ell_{i}) \in \mathcal{C}_i \text{ for } i=1,\cdots,n \} / G
\end{equation*}
appear in the Chern-Simons formulation of 3d gravity for 3-manifolds of the form $\mathcal{M}^3 \cong \mathbb{R} \times \Sigma_{g,n}$ (stationary spacetimes) and five possible Lie groups $G$. Depending on the signature (Euclidean or Lorentzian) and the sign of the cosmological constant $\Lambda_C$, the possible five Lie groups $G$ (local isometry groups of the possible spacetime models of General Relativity in three  dimensions) are
\begin{table}[h]
\begin{center}
\begin{tabular}{ |c|c|c| } 
 \hline
 $\Lambda_C$ & Euclidean & Lorentzian \\ 
\hline
 0 & $\text{SU}(2) \mathrel{\triangleright}\joinrel< \mathbb{R}^3$ & $\text{SL}(2,\mathbb{R}) \mathrel{\triangleright}\joinrel< \mathbb{R}^3$ \\ 
\hline
 $>0$ & $\text{SU}(2)\times \text{SU}(2)$  & $\text{SL}(2,\mathbb{C})_{\mathbb{R}}$ \\ 
\hline
$<0$ & $\text{SL}(2,\mathbb{C})_{\mathbb{R}}$ & $\text{SL}(2,\mathbb{R}) \times \text{SL}(2,\mathbb{R})$ \\
 \hline
\end{tabular}
\caption{\label{tab:LIG} Local isometry groups of pure 3d gravity.}
\end{center}
\end{table}

and so the possible associated Lie algebras $\mathfrak{g}$ are
\begin{table}[h]
\begin{center}
\begin{tabular}{ |c|c|c| } 
 \hline
 $\Lambda_C$ & Euclidean & Lorentzian \\ 
\hline
 0 & $\mathfrak{iso}(3)$ & $\mathfrak{iso}(2,1)$ \\ 
\hline
 $>0$ & $\mathfrak{su}(2)\oplus \mathfrak{su}(2) $  & $\mathfrak{so}(3,1)$ \\ 
\hline
$<0$ & $\mathfrak{so}(3,1)$ & $\mathfrak{sl}(2,\mathbb{R}) \oplus \mathfrak{sl}(2,\mathbb{R})$ \\
 \hline
\end{tabular}
\caption{\label{tab:LieAlg} Lie algebras of the previous local isometry groups.}
\end{center}
\end{table}

These five Lie algebras can be described in a unified  way as follows. They are six dimensional real Lie algebras generated by $\{J_0,J_1,J_2,P_0,P_1,P_2\}$ with commutation relations given by  
\begin{equation}\label{commut}
[J_a,J_b]=\epsilon_{abc}J^c, \qquad [J_a,P_b]=\epsilon_{abc}P^c, \qquad [P_a,P_c]=\lambda \epsilon_{abc}J^c
\end{equation}
where $\lambda=-c^2 \Lambda_C$ and the indices are raised or lowered using the metric 
\begin{equation*}
\eta_{ab}:= \text{diag} \left( 1,-\frac{|c|^2}{c^2}, -\frac{|c|^2}{c^2} \right)
\end{equation*}
with $c \in i\mathbb{R}$ and $c \in \mathbb{R}$ for the Euclidean and Lorentzian cases, respectively, where for the latter it is interpreted as the \textit{speed of light}. These Lie algebras are denoted by $\mathfrak{g}_\lambda$, since by picking the right metric, i.e. $\text{diag}(1,1,1)$ for the Euclidean or $\text{diag}(1,-1,-1)$ for the Lorentzian, the commutation relations depend only on the \textit{parameter} $\lambda$. Similarly,  we denote the corresponding Lie groups in Table \ref{tab:LIG}  by $G_\lambda$. \\\\
For any Lie group $G$, the Poisson structure of the moduli space $\mathcal{P}_{G,\{C_i\}}^{g,n}$ depends on the choice of a non-degenerate symmetric Ad-invariant bilinear form over the Lie algebra $\mathfrak{g}=\textit{Lie}(G)$. In  the case $G={SL}(2,\mathbb{R})$  considered  above, this dependence translates into the (equivalent) choice of the element $K$ in $(S^2 \mathfrak{sl}(2,\mathbb{R}))^{\mathfrak{sl}(2,\mathbb{R})}$. For the Lie groups we are interested in this paper $G_\lambda$, the space of non-degenerate symmetric Ad-invariant bilinear forms over the Lie algebra $\mathfrak{g}_\lambda$ has two generators \cite{WITTEN198846}
\begin{equation}\label{forms}
t(J_a,J_b)=0, \qquad t(J_a,P_a)=c^2 \eta_{ab}, \qquad t(P_a,P_b)=0,
\end{equation}
and 
\begin{equation}\label{formt}
s(J_a,J_b)=\eta_{ab}, \qquad s(J_a,P_a)=0, \qquad s(P_a,P_b)=\lambda \eta_{ab}.
\end{equation}
The bilinear form $t$ is usually called the \textit{standard (gravity) pairing}, since when it is used in the Chern-Simons formulation of three dimensional gravity the action reduces to the Einstein-Hilbert action; meanwhile $s$ is commonly referred as the \textit{exotic pairing}. Nevertheless, both actions provide the same phase space, i.e. the same equations of motions, but equipped with different symplectic/Poisson structures. \\\\
Following the same reductions presented above for $G=\text{SL}(2,\mathbb{R})$ and the four punctured 2-sphere, the canonical Poisson structure over the 
\begin{equation*}
4n + 12g - 12  = n \times (6)+2g \times (6) - \underbrace{n \times (2)}_{(i)} - \underbrace{6}_{(ii)} - \underbrace{6}_{(iii)}
\end{equation*}
dimensional space $\mathcal{P}^{g,n}_{G_\lambda,\{\mathcal{C}_i\}}$, for any of the Lie groups $G_\lambda$ presented above, can be recovered starting from the $6n+12g$ dimensional extended space 
\begin{equation*} 
\mathcal{P}_{G_\lambda, \text{ext}}^{g,n} \equiv \underbrace{(G_\lambda)_1 \times \cdots \times (G_\lambda)_n}_{\text{n-times}} \times \underbrace{(G_\lambda)_{A1} \times (G_\lambda)_{B1} \times \cdots \times (G_\lambda)_{Ag} \times (G_\lambda)_{Bg}}_{\text{2g-times}} 
\end{equation*}
equipped with the Fock-Rosly Poisson structure \cite{Fock:1998nu}
\begin{equation}\label{FR}
\begin{split}
\Pi^{g,n}_{\text{ext}}(r)= & \sum_{i=1}^{n}r^{ab} \left( \frac{1}{2}R_{a}^{i} \wedge R_b^{i} +\frac{1}{2} L_a^{i} \wedge L_b^{i} + R^{i}_a \wedge L_b^{i} \right) + \sum_{1 \leq i <j \leq n} r^{ab} (R_a^{i} + L_a^{i}) \wedge (R_b^{j}+L_b^{j}) \\
& + \sum_{1\leq i < j \leq g} r^{ab} \left (R_a^{A_i} + L_a^{A_i}+R_a^{B_i}+L_a^{B_i} \right) \wedge \left(R_b^{A_i} + L_b^{A_i}+R_b^{B_i}+L_b^{B_i} \right) \\
& + \sum_{i=1}^n \sum_{j=1}^g r^{ab} \left(R_a^{M_i} + L_a^{M_i} \right) \wedge \left(R_b^{A_j} + L_b^{A_j} + R_b^{B_j} + L_b^{B_j} \right) \\
& + \sum_{j=1}^g r^{ab} \bigg[ \frac{1}{2}(R_a^{A_j} \wedge R_b^{A_i} + L_a^{A_i} \wedge L_b^{A_i} + R_a^{B_i} \wedge R_{b}^{B_i} + L_a^{B_i} \wedge L_b^{B_i}) \\
& + R_a^{A_i} \wedge (R_b^{B_i} + L_b^{A_i} + L_{b}^{B_i}) + R_a^{B_i} \wedge (L_b^{A_i} + L_b^{B_i}) + L_a^{A_i} \wedge L_b^{B_i} \bigg], 
\end{split}
\end{equation}
where $r \in \mathfrak{g}_\lambda \otimes \mathfrak{g}_\lambda$ is a classical $r$-matrix such that $\text{Sym}(r)=K_{\alpha \beta}$, where  $K_{\alpha\beta} \in (S^2 \mathfrak{g}_\lambda)^{\mathfrak{g}_\lambda}$ is the symmetric element 
\begin{equation}\label{casimglambda}
   K_{\alpha \beta} \equiv \frac{\alpha}{\alpha^2-\lambda \beta^2}(J_a \otimes P^a + P_a \otimes J^a) - \frac{\beta}{\alpha^2-\lambda \beta^2} (\lambda J_a \otimes J^a + P_a \otimes P^a)
\end{equation}
associated to the $\mathfrak{g}_\lambda$ Ad-invariant symmetric bilinear form 
\begin{equation}\label{genbin}
(\cdot,\cdot)_{\alpha,\beta} = \alpha t(\cdot,\cdot)+\beta s(\cdot,\cdot)
\end{equation}
with $t$ and $s$ given by (\ref{forms}) and (\ref{formt}), respectively, and $\alpha,\beta \in \mathbb{R}$. The non-degeneracy of $(\cdot,\cdot)_{\alpha\beta}$ is equivalent to the condition $\alpha^2-\lambda \beta^2 \neq 0$ (see e.g \cite{Meusburger:2007ad}). \\\\
Exactly as before, given that the constraints reducing the first $n$-copies of $G_\lambda$ to conjugacy classes $\mathcal{C}_\lambda$\footnote{For simplicity we consider just conjugacy classes generated by elements of the form $e^{\mu J_0+ s P_0}$.} are Poisson functions with respect to $\Pi_{\text{ext}}^{g,n}$ (see \cite{meusburger2003poisson} again for details), we have the partially-constrained space given by

\begin{equation*} 
\mathcal{P}_{G_\lambda,\text{ext(cc)}}^{g,n} \equiv \underbrace{(\mathcal{C}_\lambda)_1 \times \cdots \times (\mathcal{C}_\lambda)_n}_{\text{n-times}} \times \underbrace{(G_\lambda)_{A1} \times (G_\lambda)_{B1} \times \cdots \times (G_\lambda)_{Ag} \times (G_\lambda)_{Bg}}_{\text{2g-times}}. 
\end{equation*}

In this case the gauge fixing procedure, via four auxiliary constraint functions defined again over the product of the first two conjugacy classes $(\mathcal{C}_\lambda)_1 \times (\mathcal{C}_\lambda)_2$, will provide the intermediate space 

\begin{equation*}
\mathcal{P}_{G_\lambda, \{\mathcal{C}_i\},GF(1,2)}^{g,n} = H_{\lambda,0} \times \underbrace{(\mathcal{C}_\lambda)_3 \times \cdots \times (\mathcal{C}_\lambda)_n}_{(n-2)\text{-times}} \times \underbrace{(G_\lambda)_{A1} \times (G_\lambda)_{B1} \times \cdots \times (G_\lambda)_{Ag} \times (G_\lambda)_{Bg}}_{2g\text{-times}}  
\end{equation*}

of dimension $4n+12g-6$, where $H_{\lambda,0}$ is the Cartan subgroup with Lie subalgebra $\mathfrak{h}_{\lambda,0}$ generated by $J_0$ and $P_0$ (see \cite{Meusburger:2012wc}). Similarly to the $\text{SL}(2,\mathbb{R})$ case, the reduced Poisson structure (i.e the Dirac brackets) will be given for $F,\tilde{F} \in C^\infty(G_\lambda^{n-2+2g})$ by 
\begin{equation}\label{gfixbra}
\begin{split}
\{\gamma,\psi\} & =0, \\
\{F,\gamma \} & = (R_{J_0}+L_{J_0})F, \\
\{F,\psi\} & = (R_{P_0}+L_{P_0})F, \\
\{F,\tilde{F}\}(\gamma,\psi) & =\Pi_\text{ext}^{g,n}(r(\gamma,\psi))(dF,d\tilde{F}),
\end{split}
\end{equation}
where $\mathfrak{h}_{\lambda,0}$ is parametrized via $\gamma J_0 + \psi P_0$ and 
\begin{equation*}
r: \mathfrak{h}_{\lambda,0}^* \to \mathfrak{g}_\lambda \otimes \mathfrak{g}_\lambda
\end{equation*}
is a $\mathfrak{h}$-equivariant solution of the CDYBE (\ref{CDYB}) with $\text{Sym}(r(\gamma,\psi))=\text{Sym}(r)=K_{\alpha \beta}$. \\\\
As shown in Appendix B, the Lie algebras $\mathfrak{g}_\lambda$ have at most four conjugacy classes of Cartan subalgebras: For $\mathfrak{so}(4)$, $\mathfrak{so}(3,1)$ and $\mathfrak{iso}(3)$ the set $\mathfrak{C}_{\mathfrak{g}_\lambda}^{\text{Ad}}$ is a singleton, being all the Cartan subalgebras conjugate to the algebra $\mathfrak{h}_{\lambda,0}$. For $\mathfrak{iso}(2,1)$ the set $\mathfrak{C}^{\text{Ad}}_{\mathfrak{g}_\lambda}$ has cardinality two, since any Cartan subalgebra is conjugated to one of the non-conjugate Cartan subalgebras $\mathfrak{h}_{\lambda,0}$ or $\mathfrak{h}_{\lambda,1}$ (generated by $J_1$ and $P_1$). Finally for $\mathfrak{so}(2,2)$, $\mathfrak{C}_{\mathfrak{g}_\lambda}^{\text{Ad}}$ has cardinality four since any Cartan subalgebra is conjugated to $\mathfrak{h}_{\lambda,0}$, $\mathfrak{h}_{\lambda,1}$, $\mathfrak{h}_{01}^{\pm}$ (generated by $J_0+P_0$ and $J_1-P_1$) or $\mathfrak{h}_{01}^{\mp}$ (generated by $J_0-P_0$ and $J_1+P_1$). Hence, depending on the Lie algebra $\mathfrak{g}_\lambda$ and the conjugacy classes considered $\{\mathcal{C}_i\}$, the Poisson structures of the gauge-fixed phase space are in correspondence with classical dynamical $r$-matrix for the triples $(\mathfrak{g}_\lambda , \mathfrak{h}_{\lambda} , K_{\alpha \beta})$  with $\mathfrak{h}_{\lambda}=\mathfrak{h}_{\lambda,0}$ and $\mathfrak{h}_{\lambda}=\mathfrak{h}_{\lambda,1}$ (and also  $\mathfrak{h}_{\lambda}=\mathfrak{h}^\pm_{01}$ and $\mathfrak{h}_{\lambda}=\mathfrak{h}^\mp_{01}$ for $\mathfrak{so}(2,2)$).

The rest of the paper is organized as follows: In Section \ref{Sec2paper1} we present a systematical treatment (following the same fashion of \cite{Osei:2017ybk}) of the CDYBE for the Lie algebras $\mathfrak{g}_\lambda$, using as a main tool the fact these Lie algebras can be realized as \textit{generalized complexifications} of $\mathfrak{su}(2)$ and $\mathfrak{sl}(2,\mathbb{R})$. In Section \ref{Sec3paper1} a full description of the set $\text{Dyn}^{\mathfrak{C}}(\mathfrak{g}_{\lambda},K_{\alpha \beta})$ and the moduli space $\mathcal{M}^{\mathfrak{C}}(\mathfrak{g}_{\lambda},K_{\alpha \beta})$ of Cartan classical dynamical $r$-matrices is presented, including also some dynamical generalizations of well-known solutions of the CYBE. Section \ref{Sec4paper1} is devoted to showing  that the classical dynamical $r$-matrices found in the previous section are gauge equivalent to a family of classical dynamical $r$-matrices studied by Feher, Gabor, Marshall, Palla and Pusztai in the setting of gauge-fixed WZNW models. Finally, in the Appendices we include some digressions regarding the action of $G$ over the sets of classical dynamical $r$-matrices, the Cartan subalgebras of the Lie algebras $\mathfrak{g}_\lambda$ and the Weierstrass factorization theorem.


\section{Structure of the  CDYBE for \texorpdfstring{$\mathfrak{g}_\lambda$}{}}\label{Sec2paper1}
\subsection{The Lie algebras \texorpdfstring{$\mathfrak{g}_\lambda$}{} as  generalised complexifications}
Here we  solve the CDYBE for the Lie algebras $\mathfrak{g}_\lambda$ (and any Lie subalgebra $\mathfrak{h}_\lambda$) by considering  the most general Ansatz in the CDYBE associated to $(\mathfrak{g}_\lambda,\mathfrak{h}_\lambda,K_{\alpha \beta})$ and solving the resulting equations.  \\\\
The realization of the Lie algebras $\mathfrak{g}_\lambda$ as the real form of a \textit{generalized complexification} of $\mathfrak{su}(2)$ or $\mathfrak{sl}(2,\mathbb{R})$, depending on the signature, has proved useful for describing the classical $r$-matrices and Poisson-Lie structures associated to the local isometry groups of 3d gravity (see e.g. \cite{Meusburger:2007ad}). The Lie algebras $\mathfrak{iso}(3)$, $\mathfrak{so}(3,1)$ and $\mathfrak{so}(4)$ can be constructed via generalized complexification of the real Lie algebra $\mathfrak{su}(2)$. Analogously, the Lie algebras $\mathfrak{iso}(2,1)$, $\mathfrak{so}(3,1)$ and $\mathfrak{so}(2,2)$ can be obtained from $\mathfrak{sl}(2,\mathbb{R})$ using the same construction. \\\\
This \textit{generalized complexification} requires the introduction of a formal parameter $\theta$ such that $\theta^2=\lambda$ and to set 
\begin{equation}\label{complex}
P_a= \theta J_a \qquad \text{for} \qquad a=0,1,2,
\end{equation}
where $\{J_a\}$ are precisely the generators of $\mathfrak{su}(2)$ or $\mathfrak{sl}(2,\mathbb{R})$, i.e. 
\begin{equation*}
[J_a,J_b]=\epsilon_{abc}J^c.
\end{equation*}
using the metrics $\text{diag}(1,1,1)$ or $\text{diag}(1,-1,-1)$ to raised or lowered indices in the $\mathfrak{su}(2)$ or $\mathfrak{sl}(2,\mathbb{R})$ cases, respectively. \\\\
Formally, the main ingredient of the construction is a ring denoted by $\mathbb{R}_\lambda$, which is obtained by adjoining a formal element $\theta$ such that $\theta^2=\lambda$. 
\begin{definition}\label{definitionRLAMBDA}
\cite{Meusburger:2006fn}\textbf{.} Let $\lambda \in \mathbb{R}$. Then we denote by
\begin{equation}\label{Rlambda}
\mathbb{R}_\lambda \equiv \{ (x+ \theta y) \ | \ x,y \in \mathbb{R}, \ \theta^2=\lambda  \}, 
\end{equation}
the ring homeomorphic to $(\mathbb{R}^2,+)$ as an abelian group and with product given by 
\begin{equation*}
    (x+ \theta y)\cdot(z + \theta w)=(xz+\lambda yw) + \theta (xw+yz).
\end{equation*}
for $x,y,z,w \in \mathbb{R}$.
\end{definition}
The construction of $\mathbb{R}_\lambda$ mimics the one of the complex numbers $\mathbb{C}$ by adjoining a formal element $i$ to $\mathbb{R}$ such that $i^2=-1$. Indeed it becomes a field in the case $\lambda<0$, isomorphic to $\mathbb{C}$. For the other cases, the rings posses zero divisors and are isomorphic to the dual and hyperbolic numbers in the cases $\lambda=0$ and $\lambda>0$, respectively. \\\\ 
Then for $\mathfrak{g}=\mathfrak{su}(2)$ or $\mathfrak{g}=\mathfrak{sl}(2,\mathbb{R})$, the \textit{generalized complexification} $\mathfrak{g} \otimes \mathbb{R}_\lambda$ is isomorphic to one of the five Lie algebras $\mathfrak{g}_\lambda$, via the identification established in (\ref{complex}): 
\begin{lemma}\label{LemmaCat} 
\cite{Meusburger:2006fn}\textbf{.} For $\lambda \in \mathbb{R}$ we have the following isomorphisms
\begin{equation*}
    \mathfrak{su}(2) \otimes \mathbb{R}_\lambda \cong \begin{cases} \mathfrak{iso}(3) & \quad \lambda=0 \\ \mathfrak{so}(4) & \quad \lambda>0 \\ \mathfrak{so}(3,1) & \quad \lambda<0  \end{cases}, \qquad \qquad \mathfrak{sl}(2,\mathbb{R}) \otimes \mathbb{R}_\lambda \cong \begin{cases} \mathfrak{iso}(2,1) & \quad \lambda=0 \\ \mathfrak{so}(2,2) & \quad \lambda>0 \\ \mathfrak{so}(3,1) & \quad \lambda<0  \end{cases}
\end{equation*}
\end{lemma}
In the next subsection the utility of this algebraic observation will become apparent, and how it allows us to split the CDYBE into simpler pieces. Also, since the CDYBE involves partial derivatives it will result useful to introduce the notion of differentiability in the ring $\mathbb{R}_\lambda$, using the fact its construction resembles the one of the complex numbers $\mathbb{C}$. \\\\   
Over the ring $\mathbb{R}_\lambda$, the \textit{generalised complexified} variable is defined as 
\begin{equation}\label{z}
z=\psi+ \theta \gamma    
\end{equation}
and its conjugate is given by 
\begin{equation}\label{conjz}
\overline{z}=\psi- \theta \gamma.    
\end{equation}
The generalized differentials of these pair of conjugate variable are defined, analogous to the complex differentials, by
\begin{equation*}
    dz=d \psi + \theta d \gamma \qquad \text{and} \qquad d \overline{z}= d \psi- \theta d \gamma.
\end{equation*}
In the cases when $\lambda \neq 0$, the generalized analogues of the complex (anti-)holomorphic derivatives are given by
\[
\partial_z= \frac 12 \left( \frac{ \partial}{\partial \psi} + \frac 1 \theta \frac {\partial}{ \partial \gamma}\right) \qquad \text{and} \qquad 
\partial_{\bar z}= \frac 12 \left( \frac{ \partial}{\partial \psi} - \frac 1 \theta \frac {\partial}{ \partial \gamma}\right),
\]
satisfying
\[
dz(\partial_z)=1, \quad d\bar z(\partial_z)=0, \quad d\bar z(\partial_{\bar z})=1, \quad d z(\partial_{\bar z})=0.
\]
Explicitly, by considering a generalized function $w=b+ \theta c: \mathbb{R}_\lambda \to \mathbb{R}_\lambda$, its generalized partial derivatives are given by
\begin{equation}\label{partcompl}
\partial_{z} w =\frac 12 \left( \frac{\partial b}{\partial \psi}+ \frac{\partial c}{\partial \gamma}  +\theta \frac{\partial c}{\partial \psi}  + \frac 1 \theta \frac{\partial b}{\partial \gamma}\right) \qquad \text{and} \qquad 
\partial_{\bar z} w  = \frac 12 \left( \frac{\partial b}{\partial \psi}- \frac{\partial c}{\partial \gamma}  +\theta \frac{\partial c}{\partial \psi}  - \frac 1 \theta \frac{\partial b}{\partial \gamma}\right).
\end{equation}
Hence, a generalized function $w$ is $\mathbb{R}_\lambda$-holomorphic, i.e
\begin{equation}\label{holo1}
\partial_{\bar z} w  =0  
\end{equation}
if and only if  
\begin{equation*}
 \frac{\partial b}{\partial \psi} = \frac{\partial c}{\partial \gamma} \quad \text{and} \quad  \frac{\partial b}{\partial \gamma} = \lambda \frac{\partial c}{\partial \psi} 
 \end{equation*}
and consequently the generalized partial derivative of an $\mathbb{R}_\lambda$-holomorphic function becomes a generalized total derivative,  given by 
 \begin{equation}\label{holo2}
\frac{d w}{\partial z} \equiv \frac{\partial b}{\partial \psi} + \theta  \frac{\partial c}{\partial \psi}      
 \end{equation}
In the next section, when finding explicitly solutions of the CDYBE for $\mathfrak{g}_\lambda$, this generalized notion of holomorphicity for $\mathbb{R}_\lambda$ will play a key role. 
\begin{remark}
We have not formally defined   (anti-)holomorphic derivatives for functions of $\mathbb{R}_\lambda$ here, and  are aware that the presence of zero-divisors makes the definition via limits of difference-quotients tricky. However, we will derive all differential equations in this paper using conventional real calculus, and use the ring $\mathbb{R}_\lambda$ as a convenient tool for collecting pairs  of real equations into one generalised complex one. 
\end{remark}
 
\subsection{Casimirs and the CDYBE for \texorpdfstring{$\mathfrak{g}_\lambda$}{}} 
In terms of the generalized complexification (\ref{complex}), the Casimir (\ref{casimglambda}) can be written as 
\begin{equation}\label{casi}
K_{\alpha\beta}  = \frac{\alpha}{\alpha^2-\lambda \beta^2}(J_a \otimes \theta J^a + \theta J_a \otimes J^a) - \frac{\beta}{\alpha^2-\lambda \beta^2} (\lambda J_a \otimes J^a + \theta J_a \otimes \theta J^a)
\end{equation}
and consequently, by extending linearly the Classical Yang-Baxter map over the ring $\mathbb{R}_\lambda$, we obtain 
\begin{equation*}
\begin{split}
\Omega_{\alpha\beta} & = \text{CYB}(K_{\alpha \beta}) \\
& = \mu \epsilon_{abc}(\lambda J^a \otimes J^b \otimes J^c + J^a \otimes \theta J^b \otimes \theta J^c + \theta J^a \otimes J^b \otimes \theta J^c + \theta J^a \otimes \theta J^b \otimes J^c) \\
& \quad + \nu \epsilon_{abc} ( \theta J^a \otimes \theta J^b \otimes \theta J^c + \lambda \theta J^a \otimes J^b \otimes J^c + \lambda J^a \otimes \theta J^b \otimes J^c + \lambda J^a \otimes J^b \otimes \theta J^c) \\
& = [\mu (\lambda \text{id} \otimes \text{id} \otimes \text{id} + \text{id} \otimes \theta \otimes \theta + \theta \otimes \text{id} \otimes \theta + \theta \otimes \theta \otimes \text{id})] \epsilon_{abc}J^a \otimes J^b \otimes J^c \\
& \quad + [\nu (\theta \otimes \theta \otimes \theta) + \nu \lambda (\theta \otimes \text{id} \otimes \text{id} + \text{id} \otimes \theta \otimes \text{id} + \text{id} \otimes \text{id} \otimes \theta)] \epsilon_{abc}J^a \otimes J^b \otimes J^c
\end{split}
\end{equation*}
where 
\begin{equation}\label{munu}
\mu \equiv \frac{\alpha^2 + \lambda \beta^2}{(\alpha^2-\lambda \beta^2)^2}, \qquad \nu \equiv -\frac{2 \alpha \beta}{(\alpha^2-\lambda \beta^2)^2}.
\end{equation}
Here it is important that,  even though  the Lie algebras $\mathfrak{g}_\lambda$ can be recovered via the generalized complexifications $\mathfrak{su}(2) \otimes \mathbb{R}_\lambda$ and $\mathfrak{sl}(2,\mathbb{R}) \otimes \mathbb{R}_\lambda$, any tensor product of $\mathfrak{g}_\lambda$ with itself is taken over $\mathbb{R}$. \\\\
Since the goal is to obtain classical dynamical $(\mathfrak{g}_\lambda,\mathfrak{h}_\lambda,K_{\alpha \beta})$ $r$-matrices, we are looking for $r \in C^1(\mathfrak{h}_\lambda^*) \otimes (\mathfrak{g}_\lambda \wedge \mathfrak{g}_\lambda)$ such that   
\begin{equation}\label{sCDYB}
\text{CDYB}(r)=-\Omega_{\alpha \beta}.
\end{equation}
Analogously to \cite{Osei:2017ybk}, the most general antisymmetric classical dynamical $r$-matrix is of the form  
\begin{equation}\label{drmatrixans}
\begin{split}
r(\bm{\gamma},\bm{\psi}) & = A_{ab}(\bm{\gamma},\bm{\psi})J^a \otimes J^b + B_{ab}(\bm{\gamma},\bm{\psi})P^a \otimes J^b - B_{ab}(\bm{\gamma},\bm{\psi})J^b \otimes P^a + C_{ab}(\bm{\gamma},\bm{\psi})P^a \otimes P^b \\
& = J^a \otimes A(\bm{\gamma},\bm{\psi})J_a + P^a \otimes B(\bm{\gamma},\bm{\psi})J_a - B(\bm{\gamma},\bm{\psi})J_a \otimes P^a + P^a \otimes C(\bm{\gamma},\bm{\psi})P_a \\
& = (\text{id} \otimes A(\bm{\gamma},\bm{\psi}) + \theta \otimes B(\bm{\gamma},\bm{\psi}) - B(\bm{\gamma},\bm{\psi}) \otimes \theta + \theta \otimes \theta C(\bm{\gamma},\bm{\psi}) )J^a \otimes J_a
\end{split}
\end{equation}
where $A,B,C \in C^1(\mathfrak{h}^*_\lambda) \otimes  \text{L}(\mathfrak{g},\mathfrak{g})$, with $\mathfrak{g}=\mathfrak{su}(2)$ for Euclidean or $\mathfrak{g}=\mathfrak{sl}(2,\mathbb{R})$ for Lorentzian signature, such that $A$ and $C$ are skew-symmetric, i.e. $A_{ab}(\bm{\gamma},\bm{\psi})=-A_{ba}(\bm{\gamma},\bm{\psi})$ and $C_{ab}(\bm{\gamma},\bm{\psi})=-C_{ba}(\bm{\gamma},\bm{\psi})$. The skew-symmetry of these linear maps of $\mathfrak{g}$, implies there exist vector functions $v(\bm{\alpha},\bm{\psi})$ and $w(\bm{\alpha},\bm{\psi})$ in $C^1(\mathfrak{h}_{\lambda}^*, \mathbb{R}^3)$ (for $\mathfrak{g}=\mathfrak{su}(2)$) or in $C^1(\mathfrak{h}_{\lambda}^*,\mathbb{M}^{1,2})$ (for $\mathfrak{g}=\mathfrak{sl}(2,\mathbb{R})$), such that    
\begin{equation*}
A(\bm{\gamma},\bm{\psi})(\cdot)=[v(\bm{\gamma},\bm{\psi}),\ \cdot \ ] \qquad \text{and} \qquad C(\bm{\gamma},\bm{\psi})(\cdot)=[w(\bm{\gamma},\bm{\psi}), \ \cdot \ ], 
\end{equation*}
where the coordinates $(\bm{\gamma},\bm{\psi})$ parametrizes the dual Lie subalgebra $\mathfrak{h}_{\lambda}^*$ via 
\begin{equation}\label{paramet}
\gamma_0 J_0^* + \gamma_1 J_1^* + \gamma_2 J_2^* + \psi_0 P_0^* +\psi_1 P_1^* +\psi_2 P_2^*.
\end{equation}
Then, by plugging this general ansatz of $r(\bm{\gamma}, \bm{\psi})$ into the CDYBE (\ref{sCDYB}) and using some algebra, we obtain the first main result of the present paper, which consists of a \textit{set of four equations} for the linear maps $A,B,C$, which must be satisfied in order to conclude (\ref{drmatrixans}) is a solution of the Classical Dynamical Yang-Baxter equation (\ref{sCDYB}).    
\begin{theorem}\label{MainTheo1}
Let $\mathfrak{h}_{\lambda}$ be a Lie subalgebra of $\mathfrak{g}_\lambda$, with $\mathfrak{h}^*_\lambda$ parametrized like in (\ref{paramet}). A function $r \in C^1(\mathfrak{h}^*_\lambda) \otimes (\mathfrak{g}_\lambda \wedge \mathfrak{g}_\lambda)$ given by 
\begin{equation}\label{gendyn}
\begin{split}
r_{\text{d}}(\bm{\psi},\bm{\gamma}) & = K_{\alpha \beta} + r(\bm{\psi},\bm{\gamma}) \\
& = \frac{\alpha}{\alpha^2-\lambda \beta^2}(J_a \otimes P^a + P_a \otimes J^a) - \frac{\beta}{\alpha^2-\lambda \beta^2} (\lambda J_a \otimes J^a + P_a \otimes P^a) + J^a \otimes A(\bm{\psi},\bm{\gamma})J_a \\
& \qquad + P^a \otimes B(\bm{\psi},\bm{\gamma})J_a - B(\bm{\psi},\bm{\gamma})J_a \otimes P^a + P^a \otimes C(\bm{\psi},\bm{\gamma})P_a 
\end{split}
\end{equation}
is a solution of the CDYBE (\ref{sCDYB}) if and only if the linear maps $A,B,C \in C^1(\mathfrak{h}_\lambda^*) \otimes \text{L}(\mathfrak{g}_\lambda, \mathfrak{g}_\lambda$), with $A=[v^A(\bm{\psi},\bm{\gamma}),\cdot \ ]$ and $C=[v^C(\bm{\psi},\bm{\gamma}),\cdot \ ]$ skew-symmetric, satisfy the following set of four equations: 
\begin{equation}\label{main1}
-\frac{1}{2} \text{tr}(A^2)+\frac{\lambda}{2}[\text{tr}(B)^2-\text{tr}(B^2)] - \text{div}_{\bm{\gamma}}(v)  =-\mu \lambda
\end{equation}

\begin{equation}\label{main2}
-\text{tr}(CB) -  \text{div}_{\bm{\psi}}(w) = -\nu
\end{equation}

\begin{equation}\label{main3}
[B-\text{tr}(B)\text{id}](B+B^t)+\frac{1}{2}[\text{tr}(B)^2-\text{tr}(B^2)]\text{id}-CA + \lambda (C^2-\frac{1}{2}\text{tr}(C^2)\text{id}) +\text{curl}_{\bm{\psi}}(B) - \text{grad}_{\bm{\gamma}}(w)=-\mu \text{id}
\end{equation}

\begin{equation}\label{main4}
-A(B+B^t)+(B^t-\text{tr}(B)\text{id})(\lambda C -A)-\text{tr}(AB)\text{id} + \text{curl}_{\bm{\gamma}}(B^t) - \text{grad}_{\bm{\psi}}(v)=-\lambda \nu \text{id}
\end{equation}

where $[\text{grad}(v)]_{ab}=(\text{grad}(v_b))_a$ and $[\text{curl}(M)]_{ab}= \text{curl}(\text{Row}_a(M))_b$, such that the subindices $\bm{\gamma}$ and $\bm{\psi}$ indicate if the differential operators are computed with respect to the variables $(\gamma_0,\gamma_1,\gamma_2)$ or $(\psi_0,\psi_1,\psi_2)$, respectively.
\end{theorem}
\begin{proof}
In \cite{Osei:2013xra}, the contraction of the CYBE with a generic element $X \otimes Y \otimes Z \in \mathfrak{g}_\lambda^{\otimes 3}$ was found to be useful to determine classical $r$-matrices of $\mathfrak{g}_\lambda$ with symmetric part $K_{\alpha\beta}$ (equations (4.11)-(4.18)). Following the same approach, the contraction
\begin{equation*}
\langle \text{alt} [ d(r(\bm{\psi},\bm{\gamma}))], X \otimes Y \otimes Z \rangle = 0
\end{equation*}
can be decompose into the following eight independent terms in $(\mathbb{R}_\lambda)^{\otimes_{\mathbb{R}} 3}$
\begin{itemize}
\item $\text{id} \otimes \text{id} \otimes \text{id}$ 
\begin{equation}\label{ididid2}
\begin{split}
0 & = \langle J^b,X \rangle \langle J^a,Y \rangle \langle \partial_{\gamma_b}(A) J_a,Z \rangle - \langle J^a,X \rangle \langle J^b,Y \rangle \langle \partial_{\gamma_b} (A)J_a,Z \rangle + \langle J_a,X \rangle \langle \partial_{\gamma_b}(A)J_a,Y \rangle \langle J^b,Z \rangle  \\
& = \langle J^b,X \rangle \langle Y,\partial_{\gamma_b}(A^t)(Z) \rangle - \langle J^b,Y \rangle \langle X,\partial_{\gamma_b}(A^t)(Z) \rangle + \langle X,\partial_{\gamma_b}(A^t)(Y) \rangle \langle J^b,Z \rangle  \\
& = \langle \partial_{\gamma_a}(A)([J^a,[Y,X]]),Z  \rangle + \langle \langle X,\partial_{\gamma_a}(A^t)(Y) \rangle J^a,Z \rangle 
\end{split} 
\end{equation}

\item $\text{id} \otimes \theta \otimes \theta$
\begin{equation}
\begin{split}
0 & = \langle \partial_{\psi_b}(B) J^a,X  \rangle \langle J^b,Y \rangle \langle J_a,Z \rangle - \langle \partial_{\psi_b}(B) J^a,X  \rangle \langle J_a,Y \rangle \langle J^b,Z \rangle + \langle J^b,X \rangle \langle J^a,Y \rangle \langle \partial_{\gamma_b}(C) J_a,Z \rangle \\
& = \langle \partial_{\psi_b}(B^t) X , Z  \rangle \langle J^i,Y \rangle - \langle Y, \partial_{\psi_b}(B^t) X  \rangle \langle J^b,Z \rangle + \langle J^b,X  \rangle \langle Y,\partial_{\gamma_b}(C^t) Z \rangle \\
& = \langle [Y,[\partial_{\psi_b}(B^t) X,J^b]],Z \rangle + \langle \langle J^b,X \rangle \partial_{\gamma_b}(C) Y,Z \rangle
\end{split}
\end{equation}

\item $\theta \otimes \text{id} \otimes \theta$ 
\begin{equation}
\begin{split}
0 & = - \langle J^b,X  \rangle \langle \partial_{\psi_b}(B) J^a,Y \rangle \langle J_a,Z \rangle + \langle J^a,X \rangle \langle \partial_{\psi_b}(B) J_a , Y \rangle \langle J^b,Z \rangle - \langle J^a,X \rangle \langle J^b,Y \rangle \langle \partial_{\gamma_b}(C) J_a , Z \rangle \\
& = -\langle J^b,X \rangle \langle \partial_{\psi_b}(B^t) Y,Z \rangle + \langle X,\partial_{\psi_b}(B^t) Y \rangle \langle J^b,Z \rangle - \langle X,\partial_{\gamma_b}(C^t) Z \rangle \langle J^b,Y \rangle \\
& = \langle [X,[J^a,\partial_{\psi_a}(B^t) Y]],Z \rangle - \langle \langle J^b ,Y \rangle \partial_{\gamma_a}(C) X, Z \rangle 
\end{split}
\end{equation}

\item $\theta \otimes \theta \otimes \text{id}$
\begin{equation}
\begin{split}
0 & =  \langle J^b,X  \rangle \langle J^a,Y \rangle \langle \partial_{\psi_b}(B)J_a,Z \rangle - \langle J^a,X \rangle \langle J^b,Y \rangle \langle \partial_{\psi_b}(B) J_a, Z\rangle + \langle J^a,X \rangle \langle \partial_{\gamma_b}(C) J_a,Y \rangle \langle J^b,Z \rangle \\
& = \langle J^a,X \rangle \langle Y,\partial_{\psi_a}(B^t) Z \rangle + \langle J^a,Y \rangle \langle X, \partial_{\psi_a}(B^t) Z \rangle - \langle X,\partial_{\gamma_a}(C^t) Y \rangle \langle J^a,Z \rangle \\
& = \langle \partial_{\psi_a}(B) ([J^b,[Y,X]]),Z \rangle - \langle \langle X,\partial_{\gamma_a}(C^t) Y \rangle J^a, Z \rangle 
\end{split}
\end{equation}

\item $\theta \otimes \theta \otimes \theta$ 
\begin{equation}
\begin{split}
0 & =  \langle J^b,X  \rangle \langle J^a,Y \rangle \langle \partial_{\psi_b}(C) J_a,Z \rangle - \langle J^a,X \rangle \langle J^b,Y \rangle \langle \partial_{\psi_b}(C) J_a,Z \rangle + \langle J^a,X \rangle \langle \partial_{\psi_a}(C) J_a,Y \rangle \langle J^a,Z \rangle \\
& = \langle J^b,X \rangle \langle Y,\partial_{\psi_b}(C^t) Z \rangle - \langle J^b,Y \rangle \langle X,\partial_{\psi_b}(C^t) Z \rangle + \langle \langle X, \partial_{\psi_b}(C^t) Y \rangle J^b,Z \rangle \\
& = \langle \partial_{\psi_a}(C)([J^a,[Y,X]]), Z \rangle + \langle \langle X, \partial_{\psi_a}(C^t) Y \rangle J^a,Z \rangle 
\end{split}
\end{equation}

\item $\theta \otimes \text{id} \otimes \text{id}$ 
\begin{equation}
\begin{split}
0 & =  \langle J^b,X  \rangle \langle J^a,Y \rangle \langle \partial_{\psi_b}(A) J_a,Z \rangle - \langle J^a,X \rangle \langle J^b,Y \rangle \langle \partial_{\gamma_b}(B) J_a,Z \rangle + \langle J^a,X \rangle \langle \partial_{\gamma_b}(B) J_a,Y \rangle \langle J^b,Z \rangle \\
& = \langle J^a,X \rangle \langle Y,\partial_{\psi_a}(A^t) Z \rangle + \langle J^a,Y \rangle \langle X,\partial_{\gamma_a}(B^t) Z \rangle - \langle X,\partial_{\gamma_a}(B^t) Y \rangle \langle J^a,Z \rangle \\
& = \langle \langle J^a,X \rangle \partial_{\psi_a}(A) Y,Z \rangle - \langle [Y,[J^a,\partial_{\gamma_a}(B) X]],Z \rangle 
\end{split}
\end{equation}

\item $\text{id} \otimes \theta \otimes \text{id}$ 
\begin{equation}
\begin{split}
0 & = \langle J^a,X  \rangle \langle J^b,Y \rangle \langle \partial_{\psi_b}(A) J_a,Z \rangle + \langle J^b,X \rangle \langle J^a,Y \rangle \langle \partial_{\gamma_b}(B) J_a,Z \rangle - \langle \partial_{\gamma_b}(B) J^a,X \rangle \langle J_a,Y \rangle \langle J^b,Z \rangle \\
& = -\langle X,\partial_{\psi_a}(A^t) Z \rangle \langle J^a,Y \rangle + \langle \langle J^a,X \rangle Y,\partial_{\gamma_a}(B^t) Z\rangle - \langle Y,\partial_{\gamma_a}(B^t) X \rangle \langle J^a,Z \rangle \\
& = \langle [X,[\partial_{\gamma_a}(B) Y,J^a]],Z \rangle - \langle \langle J^a,Y \rangle \partial_{\psi_a}(A) X,Z \rangle 
\end{split}
\end{equation}

\item $\text{id} \otimes \text{id} \otimes \theta$ 
\begin{equation}\label{ididt2}
\begin{split}
0 & =  \langle J^a,X  \rangle \langle \partial_{\psi_b}(A) J_a,Y \rangle \langle J^b,Z \rangle - \langle J^b,X \rangle \langle \partial_{\gamma_b}(B) J^a,Y \rangle \langle J_a,Z \rangle + \langle \partial_{\gamma_b}(B) J^a,X \rangle \langle J^b,Y \rangle \langle J_a,Z \rangle \\
& = \langle X,\partial_{\psi_a}(A^t) Y \rangle \langle J^a,Z \rangle - \langle J^a,X \rangle \langle \partial_{\gamma_a}(B^t) Y,Z \rangle + \langle \partial_{\gamma_a}(B^t) X,Z \rangle \langle J^a,Y \rangle \\
& = \langle \partial_{\gamma_a}(B^t)([J^a,[Y,X]]), Z \rangle + \langle \langle X, \partial_{\psi_a}(A^t) Y \rangle J^a,Z \rangle
\end{split}
\end{equation}
\end{itemize}
Hence combining the corresponding terms among (4.11)-(4.18) in \cite{Osei:2017ybk} and (\ref{ididid2})-(\ref{ididt2}), we conclude the contraction 
\begin{equation}
\langle \text{CDYBE}(r(\bm{\psi},\bm{\gamma})), X \otimes Y \otimes Z \rangle = - \langle \Omega_{\alpha \beta}, X \otimes Y \otimes Z \rangle
\end{equation}
can be decomposed into a set of eight terms, which reduces to the following of \textit{four independent coupled partial differential equations} for the linear maps $A,B,C$
\begin{equation*}
\begin{split}
& [A(X),A(Y)]-A^2([X,Y])+\lambda (B([X,B^t(Y)]+[B^t(X),Y])+[B^t(X),B^t(Y)]) \\
& \qquad \qquad \qquad \qquad \qquad \qquad \qquad \qquad \qquad \quad - \partial_{\gamma_a}(A)([J^a,[X,Y]]) + \langle \partial_{\gamma_a}(A)(X),Y \rangle J^a  = - \mu \lambda [X,Y]
\end{split}
\end{equation*}

\begin{equation*}
\begin{split}
& B^tC([X,Y])+[B(X),C(Y)] + [C(X),B(Y)] - C([B(X),Y]+[X,B(Y)]) \\
& \qquad \qquad \qquad \qquad \qquad \qquad \qquad \qquad \qquad \quad - \partial_{\psi_a}(C)([J^a,[X,Y]]) + \langle \partial_{\psi_a}(C)(X), Y \rangle J^a =-\nu [X,Y]
\end{split}
\end{equation*}

\begin{equation*}
\begin{split}
& [B(X),B(Y)]-B([X,B(Y)]+[X,B(Y)])-AC([X,Y])+\lambda[C(X),C(Y)] \\
& \qquad \qquad \qquad \qquad \qquad \qquad \qquad \qquad \qquad \quad - \partial_{\psi_a}(B) ([J^a,[X,Y]]) + \langle \partial_{\gamma_a}(C)(X), Y \rangle J^a = -\mu [X,Y]
\end{split}
\end{equation*}

\begin{equation*}
\begin{split}
& B^t A([X,Y])-[B^t(X),A(Y)]-[A(X),B^t(Y)]+\lambda C([X,B^t(Y)]+[B^t(X),Y])\\
& \qquad \qquad \qquad \qquad \qquad \qquad \qquad \qquad \qquad \quad - \partial_{\gamma_a}(B^t) ([J^a,[X,Y]]) + \langle \partial_{\psi_a}(A)(X), Y \rangle J^a =-\lambda \nu [X,Y] 
\end{split}
\end{equation*}
Finally, using adjugates (exactly like in \cite{Osei:2013xra} and \cite{Osei:2017ybk}) we derive the 4 equations (\ref{main1})-(\ref{main4}).
\end{proof}
The equations (\ref{main1})-(\ref{main4}) are the \textit{dynamical generalization} of (4.39) in \cite{Osei:2017ybk}. It is immediate to notice that if we assume the linear maps are constant, i.e. $A,B,C \in \text{L}
(\mathfrak{g}_\lambda,\mathfrak{g}_\lambda) \subset C^1(\mathfrak{h}_\lambda^*) \otimes \text{L}(\mathfrak{g}_\lambda,\mathfrak{g}_\lambda)$, then the former set of equations reduces to the later, as expected. 

\section{Classical dynamical \texorpdfstring{$r$}{}-matrices for \texorpdfstring{$\mathfrak {g}_\lambda$}{}}\label{Sec3paper1}



\subsection{Classical dynamical \texorpdfstring{$(\mathfrak{g}_\lambda,\mathfrak{h}_\lambda,K_{\alpha \beta})$}{} \texorpdfstring{$r$}{}-matrices}
In the previous theorem we find a set of equations (\ref{main1})-(\ref{main4}), equivalent to the CDYBE (\ref{sCDYB}) for any Lie subalgebra $\mathfrak{h}_\lambda$ of $\mathfrak{g}_\lambda$, nevertheless, as mentioned in the introduction we are interested in the case when it is a Cartan subalgebra. Since any Cartan Lie subalgebra of $\mathfrak{g}_\lambda$ is conjugate-equivalent to $\mathfrak{h}_{\lambda,0}$ and/or $\mathfrak{h}_{\lambda,1}$  (also to $\mathfrak{h}^{\pm}_{\lambda,01}$ or  $\mathfrak{h}^{\mp}_{\lambda,01}$ only for $\mathfrak{so}(2,2)$), for our purposes it is enough to determine the set $\text{Dyn}(\mathfrak{g}_\lambda,\mathfrak{h}_\lambda,K_{\alpha \beta})$ for the cases where $\mathfrak{h}_\lambda$ is $\mathfrak{h}_{\lambda,0}$ and $\mathfrak{h}_{\lambda,1}$ (and also $\mathfrak{h}^{\pm}_{01}$ and $\mathfrak{h}^{\pm}_{10}$ for $\mathfrak{so}(2,2)$). This gives us the full moduli space $\mathcal{M}^{\mathfrak{C}}(\mathfrak{g}_\lambda,K_{\alpha \beta})$ and then, via dynamical gauge transformations, the full space $\text{Dyn}^{\mathfrak{C}}(\mathfrak{g}_\lambda,K_{\alpha \beta})$ can be generated. 

For all $\lambda \in \mathbb{R}$, the Cartan subalgebras $\mathfrak{h}_\lambda$ are Abelian subalgebras of $\mathfrak{g}_\lambda$. Therefore  the $\mathfrak{h}_\lambda$-equivariance condition (\ref{equiv}) of a classical dynamical $r$-matrix reduces to 
\begin{equation}\label{equivabel}
[r_\text{d}(x),h \otimes 1 + 1 \otimes h]=0
\end{equation}
for all $x \in \mathfrak{h}_{\lambda}^*$ and $h \in \mathfrak{h}_{\lambda}$. \\\\
Using  the notation introduced  for the antisymmetric part $r$ of $r_\text{d}$ in (\ref{gendyn}), we now determine  the implications of the $\mathfrak{h}_\lambda$-equivariance condition
for the linear maps $[v,\cdot]$, $B$ and $[w,\cdot]$ in \eqref{drmatrixans}, which we write more explicitly as 
\begin{equation}\label{rdynsa}
r(\psi,\gamma)  = \epsilon_{abc}v^c(\psi,\gamma) J^a \otimes J^b + B_{ab}(\psi,\gamma)(P^a \otimes J^b - J^b \otimes P^a) + \epsilon_{abc}w^c(\psi,\gamma) P^a \otimes P^b.
\end{equation}
As we shall see shortly,  this condition only has interesting solutions for the  Cartan subalgebras $\mathfrak{h}_{\lambda,0}$ and $\mathfrak{h}_{\lambda,1}$. 
In  the cases  $\mathfrak{h}_\lambda=\mathfrak{h}_{01}^{\pm}$ (arising in  $\mathfrak{g}_\lambda=\mathfrak{so}(2,2)$), the equivariance condition implies that all the coefficients in (\ref{rdynsa}) vanish except for $B_{01}$, $B_{10}$, $v^2$ and $w^2$, and additionally $B_{01}=B_{10}=v^2=-w^2$. This means that the most general $\mathfrak{h}^{\pm}_{01}$-invariant element in $\mathcal{F}(\mathfrak{h}^{\pm *}_{01}) \otimes \mathfrak{so}(2,2) \otimes \mathfrak{so}(2,2)$ is of the form 
\begin{equation}\label{h01plusminus}
    K_{\alpha \beta} + f(\psi,\gamma)(J_0 + P_0) \wedge (J_1-P_1)
\end{equation}
with $f \in \mathcal{F}(\mathfrak{h}^{\pm *}_{01})$ arbitrary. Analogously the most general $\mathfrak{h}^{\mp}_{01}$-invariant element in $\mathcal{F}(\mathfrak{h}^{\mp *}_{01}) \otimes \mathfrak{so}(2,2) \otimes \mathfrak{so}(2,2)$ is of the form 
\begin{equation}\label{h01minusplus}
    K_{\alpha \beta} + g(\psi,\gamma)(J_0 - P_0) \wedge (J_1+P_1)
\end{equation}
with $g \in \mathcal{F}(\mathfrak{h}^{\mp *}_{01})$ arbitrary.
However, neither of these lead to interesting solutions of the CDYBE.
 By plugging (\ref{h01plusminus}) or (\ref{h01minusplus}) into the CDYBE, the equations (\ref{main1}) and (\ref{main2}) reduce in both cases to $0=\mu \lambda$ and $0= \nu$, respectively, showing that $\text{Dyn}(\mathfrak{so}(2,2),\mathfrak{h}_\lambda,K_{\alpha \beta})= \emptyset$ for $\mathfrak{h}_\lambda=\mathfrak{h}^\pm_{01}$ and $\mathfrak{h}_\lambda=\mathfrak{h}^\mp_{01}$. 
Therefore, the discussion of $\text{Dyn}^{\mathfrak{C}}(\mathfrak{g}_\lambda,K_{\alpha \beta})$ is reduced to the determination of $\text{Dyn}(\mathfrak{g}_\lambda,\mathfrak{h}_{\lambda,0},K_{\alpha \beta})$ and $\text{Dyn}(\mathfrak{g}_\lambda,\mathfrak{h}_{\lambda,1},K_{\alpha \beta})$, to which we now turn.


\begin{lemma}\label{Lemma2} 
The most general $\mathfrak{h}_{\lambda,0}$-equivariant element $r \in \mathcal{F}(\mathfrak{h}_{\lambda,0}^*) \otimes (\mathfrak{g}_\lambda \otimes \mathfrak{g}_\lambda )$ is of the form
\begin{equation}\label{hequivdyn0}
r(\psi_0,\gamma_0)= K_{\alpha \beta} +     c(\psi_0,\gamma_0)(\lambda J^1 \wedge J^2 + P^1 \wedge P^2) + f(\psi_0,\gamma_0)(P^0 \wedge J^0)+b(\psi_0,\gamma_0)(P^1 \wedge J^2 - J^1 \wedge P^2)
\end{equation}
with $b,c,f \in \mathcal{F}(\mathfrak{h}_{\lambda,0}^*)$ arbitrary functions. \\\\
Similarly, the most general $\mathfrak{h}_{\lambda,1}$-equivariant element $r \in \mathcal{F}(\mathfrak{h}_{\lambda,1}^*) \otimes (\mathfrak{g}_\lambda \otimes \mathfrak{g}_\lambda )$ is of the form
\begin{equation}\label{hequivdyn1}
r(\psi_1,\gamma_1)= K_{\alpha \beta} +     c(\psi_1,\gamma_1)(\lambda J^2 \wedge J^0 + P^2 \wedge P^0) + f(\psi_1,\gamma_1)(P^1 \wedge J^1)+b(\psi_1,\gamma_1)(P^2 \wedge J^0 - J^2 \wedge P^0)
\end{equation}
with $b,c,f \in \mathcal{F}(\mathfrak{h}_{\lambda,1}^*)$ arbitrary. 
\end{lemma}
\begin{proof}
For the $\mathfrak{h}_\lambda= \mathfrak{h}_{\lambda,0}$ case, by direct computation we have for the first term in the skew-symmetric part of (\ref{rdynsa})
\begin{equation*}
[J_0 \otimes 1+ 1 \otimes J_0, \epsilon_{abc}v^c J^a \otimes J^b] = v^i (J^0 \wedge J^i),
   \end{equation*}
while for the second
\begin{equation*}
\begin{split}
[J_0 \otimes 1 + 1 \otimes J_0, B_{ab}(P^a \otimes J^b - J^b \otimes P^a)] & =  \epsilon_{ij}B_{i0}(P^j \wedge J^0) + \epsilon_{ij}B_{0i}(P^0 \wedge J^j) \\
& \quad + (\epsilon_{ij}B_{ik}+ \epsilon_{\ell k}B_{j \ell})(P^j \wedge J^k),
\end{split}
\end{equation*}
and for the third
\begin{equation*}
\begin{split}
[J_0 \otimes 1 +1 \otimes J_0, \epsilon_{abc}w^c P^a \otimes P^b] & = w^i(P^0 \wedge P^i).
\end{split}
\end{equation*}
Thus, in order to have equivariance with respect to the generator $J_0$ of $\mathfrak{h}_{\lambda,0}$, we need 
\begin{equation}
0=v^i (J^0 \wedge J^i) +  \epsilon_{ij}B_{i0}(P^j \wedge J^0) + \epsilon_{ij}B_{0i}(P^0 \wedge J^j) + (\epsilon_{i j}B_{ i k}+ \epsilon_{\ell k}B_{j \ell})(P^j \wedge J^k) + w^i(P^0 \wedge P^i)  
\end{equation}
where the indices $i,j$ just take the values $1$ or $2$. Therefore, the equivariance with respect to $J_0$ requires:  
\begin{itemize}
    \item $v^i=w^i=0$ for $i=1,2$,
    \item $B_{i0}=B_{0i}=0$ for $i=1,2$,
    \item $B_{12}=-B_{21}$, 
    \item $B_{11}=B_{22}$.
\end{itemize}
Analogously, by straightforward computation we get for the first term in the skew-symmetric part of (\ref{rdynsa})
\begin{equation*}
  [P_0 \otimes 1 + 1 \otimes P_0, \epsilon_{abc}v^c J^a \otimes J^b] = v^0(P^1 \wedge J^1 + P^2 \wedge J^2) + v^1(J^0 \wedge P^1) + v^2 (J^0 \wedge P^2),
\end{equation*}
for the second
\begin{equation*}
[P_0 \otimes 1 + 1 \otimes P_0, B_{ab}(P^a \otimes J^b - J^b \otimes P^a)] = \epsilon_{ji}B_{0i}P^j \wedge P^0 + (B_{11}+B_{22})P^1 \wedge P^2 + \lambda \epsilon_{ij}B_{ib}J^j \wedge J^b,
\end{equation*}
and for the last one
\begin{equation*}
[P_0 \otimes 1 +1 \otimes P_0, \epsilon_{abc}w^c P^a \otimes P^b] = \lambda w^0 (J^1 \wedge P^1 + J^2 \wedge P^2) + \lambda w^1 (P^0 \wedge J^1) +\lambda w^2 (P^0 \wedge J^2) 
\end{equation*}
So, the equivariance with respect to the generator $P_0$ is equivalent to the equation
\begin{equation}
\begin{split}
0 & = (\lambda w^0 - v^0)(J^1 \wedge P^1 + J^2 \wedge P^2) + (\lambda w^i -v^i)(J^0 \wedge P^i) + \epsilon_{ji}B_{0i}P^j \wedge P^0 \\
& \quad   + (B_{11}+B_{22})P^1 \wedge P^2 + \lambda \epsilon_{ij}B_{ib}J^j \wedge J^b  
\end{split}
\end{equation}
where again the indices $i$, $j$ run from $1$ to $2$, while $b$ from $0$ to $2$. This equation translates into the following set of additional conditions: 
\begin{itemize}
    \item $v^a = \lambda w^a$ for $a=0,1,2$,
    \item $B_{11}+B_{22}=0$.
\end{itemize}
Consequently, by incorporating the set of constraints derived above and considering the $\mathfrak{g}_\lambda$-invariance of $K_{\alpha \beta}$, we conclude that an element $r \in \mathcal{F}(\mathfrak{h}_{\lambda,0}^*) \otimes (\mathfrak{g}_\lambda \otimes \mathfrak{g}_\lambda)$ results to be  $\mathfrak{h}_{\lambda,0}$-equivariant if and only if it is of the form (\ref{hequivdyn0}), with the identification 
\begin{equation*}
  f \equiv B_{00}, \qquad b \equiv B_{12},  \qquad c \equiv w^0. 
\end{equation*} 
The conditions for the equivariance for the case $\mathfrak{h}_\lambda=\mathfrak{h}_{\lambda,1}$ are derived in an completely analogous way, concluding an element $r \in \mathcal{F}(\mathfrak{h}_{\lambda,1}^*) \otimes (\mathfrak{g}_\lambda \otimes \mathfrak{g}_\lambda)$ results to be  $\mathfrak{h}_{\lambda,1}$-equivariant if and only if it is of the form (\ref{hequivdyn1}), but now with the identifications $f \equiv B_{11}$, $b \equiv B_{20}$ and $c \equiv w^1$. 
\end{proof}

Having at hand the form of the most general $\mathfrak{h}_\lambda$-equivariant element of $\mathcal{C}^1(\mathfrak{h}_\lambda^*) \otimes (\mathfrak{g}_\lambda \otimes \mathfrak{g}_\lambda)$, for both $\mathfrak{h}_\lambda=\mathfrak{h}_{\lambda,0}$ (\ref{hequivdyn0}) and $\mathfrak{h}_\lambda=\mathfrak{h}_{\lambda,1}$ (\ref{hequivdyn1}), we proceed to determine all the classical dynamical $(\mathfrak{g}_\lambda,\mathfrak{h}_\lambda,K_{\alpha \beta})$ $r$-matrices by solving the CDYBE (\ref{CDYB}) using it now as Ansatz, deriving in this way the second main result of this paper.  
\begin{theorem}\label{MainTheo2}
Let $(\psi_C,\gamma_C) \in \mathbb{R}^2$ constants. The classical dynamical $(\mathfrak{g}_\lambda,\mathfrak{h}_{\lambda},K_{\alpha \beta})$ $r$-matrices for $\mathfrak{h}_{\lambda}= \mathfrak{h}_{\lambda,0}$ and $\mathfrak{h}_{\lambda}= \mathfrak{h}_{\lambda,1}$ are necessarily elements in $C^1(\mathfrak{h}_\lambda^*) \otimes (\mathfrak{g}_\lambda \otimes \mathfrak{g}_\lambda)$ of the form  
\begin{equation*}
r_d(\psi_0,\gamma_0)=K_{\alpha \beta} +  f(\psi_0,\gamma_0)(P^0 \wedge J^0)+b(\psi_0,\gamma_0) (P^1 \wedge J^2 - P^2 \wedge J^1) + c(\psi_0,\gamma_0)(\lambda J^1 \wedge J^2 + P^1 \wedge P^2) 
\end{equation*}
and 
\begin{equation*}
r_d(\psi_1,\gamma_1)= K_{\alpha \beta} + f(\psi_1,\gamma_1)(P^1 \wedge J^1)+b(\psi_1,\gamma_1)(P^2 \wedge J^0 - J^2 \wedge P^0) + c(\psi_1,\gamma_1)(\lambda J^2 \wedge J^0 + P^2 \wedge P^0),
\end{equation*}
respectively, such that $f \in C^1(\mathfrak{h}_\lambda^*)$,
\begin{equation}\label{bcoeflong}
\begin{split}
    b(\psi,\gamma) & = \frac{\alpha}{\alpha^2-\lambda \beta^2} B \left( \frac{\alpha(\psi-\psi_C)-\beta(\gamma-\gamma_C)}{\alpha^2-\lambda \beta^2}, \frac{\alpha(\gamma-\gamma_C)-\lambda \beta (\psi-\psi_C)}{\alpha^2 - \lambda \beta^2} \right) \\
    & \quad -\lambda \frac{\beta}{\alpha^2-\lambda \beta^2} C \left( \frac{\alpha(\psi-\psi_C)-\beta(\gamma-\gamma_C)}{\alpha^2-\lambda \beta^2}, \frac{\alpha(\gamma-\gamma_C)-\lambda \beta (\psi-\psi_C)}{\alpha^2 - \lambda \beta^2} \right)
\end{split}
\end{equation}
and 
\begin{equation}\label{ccoeflong}
\begin{split}
    c(\psi,\gamma) & = \frac{\alpha}{\alpha^2-\lambda \beta^2} C \left( \frac{\alpha(\psi-\psi_C)-\beta(\gamma-\gamma_C)}{\alpha^2-\lambda \beta^2}, \frac{\alpha(\gamma-\gamma_C)-\lambda \beta (\psi-\psi_C)}{\alpha^2 - \lambda \beta^2} \right) \\ 
    & \quad - \frac{\beta}{\alpha^2-\lambda \beta^2} B \left( \frac{\alpha(\psi-\psi_C)-\beta(\gamma-\gamma_C)}{\alpha^2-\lambda \beta^2}, \frac{\alpha(\gamma-\gamma_C)-\lambda \beta (\psi-\psi_C)}{\alpha^2 - \lambda \beta^2} \right)
    \end{split}
\end{equation}
where for $\mathfrak{h}_{\lambda,0}$
\begin{equation}\label{coeffbc0} 
B(\Psi,\Gamma)=
\begin{cases}
\frac{\sin(2 \Psi)}{\cos (2 \Psi) + \cosh \left(2\sqrt{|\lambda|}\Gamma \right)} \\
 \tan \left(\Psi \right) \\
\frac{\sin(2 \Psi)}{\cos(2\Psi)+\cos \left(2\sqrt{\lambda}\Gamma \right)} 
\end{cases} \quad
C(\Psi,\Gamma)=
\begin{cases}
\frac{1}{\sqrt{|\lambda|}}\frac{\sinh \left(2\sqrt{|\lambda|}\Gamma \right)}{\cos (2\Psi) + \cosh \left(2 \sqrt{|\Lambda|} \Gamma \right)} \qquad & \lambda<0 \\
\frac{\Gamma}{2 \cos^2 \left(\Psi \right)} \qquad & \lambda=0 \\
\frac{1}{\sqrt{\lambda}} \frac{\sin \left(2 \sqrt{\lambda}\Gamma  \right)}{\cos (2 \Psi)+\cos(2 \sqrt{\lambda} \Gamma)} \qquad & \lambda>0
\end{cases},
\end{equation}
while for $\mathfrak{h}_{\lambda,1}$, we have two types of solutions:
\begin{enumerate}[(i)]
\item Non-constant coefficients:   
\begin{equation}\label{coeffbc1} 
B(\Psi,\Gamma)=
\begin{cases}
\frac{\sinh(2 \Psi)}{\cosh (2 \Psi) + \cos \left(2\sqrt{|\lambda|}\Gamma \right)} \\
 \tanh \left(\Psi \right) \\
\frac{\sinh(2 \Psi)}{\cosh(2\Psi)+\cosh \left(2\sqrt{\lambda}\Gamma \right)} 
\end{cases}
\qquad C(\Psi,\Gamma)=
\begin{cases}
\frac{1}{\sqrt{|\lambda|}}\frac{\sin \left(2\sqrt{|\lambda|}\Gamma \right)}{\cosh (2\Psi) + \cos \left(2 \sqrt{|\Lambda|} \Gamma \right)}, \qquad & \lambda<0 \\
\frac{\Gamma}{2 \cosh^2 \left(\Psi \right)}, \qquad & \lambda=0 \\
\frac{1}{\sqrt{\lambda}} \frac{\sinh \left(2 \sqrt{\lambda}\Gamma  \right)}{\cosh (2 \Psi)+\cosh(2 \sqrt{\lambda} \Gamma)}, \qquad & \lambda>0
\end{cases}
\end{equation}
\item or constant coefficients  
\begin{equation}\label{conscoeff}
    B(\Psi,\Gamma)= \pm 1 \quad \text{and} \quad C(\Psi,\Gamma)=0
\end{equation}
\end{enumerate}
\end{theorem}
\begin{proof}
By plugging (\ref{hequivdyn0}) in (\ref{main1})-(\ref{main4}) we conclude that classical dynamical ($\mathfrak{g}_\lambda,\mathfrak{h}_{\lambda,0},K_{\alpha \beta}$) $r$-matrices are elements in $C^1(\mathfrak{h}_{\lambda,0}^*) \otimes (\mathfrak{g}_\lambda \otimes \mathfrak{g}_\lambda)$ of the form  
\begin{equation*}
K_{\alpha \beta} + f(\psi_0,\gamma_0)(P^0 \wedge J^0) + b(\psi_0,\gamma_0) (P^1 \wedge J^2 - P^2 \wedge J^1) +c(\psi_0,\gamma_0)(\lambda J^1 \wedge J^2 + P^1 \wedge P^2)
\end{equation*}
where the coefficients satisfy the equations 
\begin{subequations}\label{CYBEFIN}
  \begin{align}
    \lambda c^2 + b^2 - \frac{\partial c }{\partial \gamma_0} & =-\mu , \label{eq:1} \\%
    2bc - \frac{\partial c}{\partial \psi_0} & =-\nu, \label{eq:2} \\
    \lambda c^2 + b^2 - \frac{\partial b}{\partial \psi_0} & =-\mu, \label{eq:3} \\
    \lambda \frac{\partial c}{\partial \psi_0} & = \frac{\partial b}{\partial \gamma_0}.
  \end{align}
\end{subequations}
These equations can be understood as the $\mathfrak{h}_{\lambda,0}$-equivariant reduction of the set (\ref{main1})-(\ref{main4}) for the $\mathfrak{h}_\lambda= \mathfrak{h}_{\lambda,0}$ case, and can be rewritten in the more compact way
\begin{equation}\label{DYBEred0}
\begin{split}
    \frac{\partial b}{\partial \psi_0} & = \lambda c^2 + b^2 + \mu=  \frac{\partial c}{\partial \gamma_0}, \\
    \frac{\partial b}{\partial \gamma_0} &= \lambda (2bc + \nu) = \lambda \frac{\partial c}{\partial \psi_0}.
    \end{split}
\end{equation}
Analogously, by inserting (\ref{hequivdyn1}) in (\ref{main1})-(\ref{main4}) we get the classical dynamical $r$-matrices associated to $(\mathfrak{g}_\lambda,\mathfrak{h}_{\lambda,1},K_{\alpha \beta})$ have the form 
\begin{equation*}
K_{\alpha \beta} + f(\psi_1,\gamma_1)(P^1 \wedge J^1)+b(\psi_1,\gamma_1)(P^2 \wedge J^0 - J^2 \wedge P^0) + c(\psi_1,\gamma_1)(\lambda J^2 \wedge J^0 + P^2 \wedge P^0),
\end{equation*}
such that
\begin{subequations}
  \begin{align}
    -\lambda c^2 - b^2 - \frac{\partial c }{\partial \gamma_1} & =-\mu , \label{eqL:1} \\%
    -2bc - \frac{\partial c}{\partial \psi_1} & =-\nu, \label{eqL:2} \\
    -\lambda c^2 - b^2 - \frac{\partial b}{\partial \psi_1} & =-\mu, \label{eqL:3} \\
    \lambda \frac{\partial c}{\partial \psi_1} & = \frac{\partial b}{\partial \gamma_1}
  \end{align}
\end{subequations}
or more compactly,
\begin{equation}\label{DYBEred1}
\begin{split}
    \frac{\partial b}{\partial \psi_1} & = -(\lambda c^2 + b^2) + \mu  =  \frac{\partial c}{\partial \gamma_1}, \\
    \frac{\partial b}{\partial \gamma_1} & = \lambda (-2bc +   \nu) = \lambda \frac{\partial c}{\partial \psi_1}.
    \end{split}
\end{equation}
If we consider the generalized complexified variable $z_0=\psi_0 + \theta \gamma_0$ (over $\mathbb{R}_\lambda$) in the sense of (\ref{z}) and construct the function over $\mathbb{R}_\lambda$ given by 
\begin{equation*}
    w(z_0)= b(z_0)+ \theta c(z_0),
\end{equation*}
then from (\ref{partcompl}), we deduce the equations (\ref{DYBEred0}) can be re-expressed over the ring $\mathbb{R}_\lambda$ as 
\begin{subequations}
\begin{align}
\partial_{\overline{z}_0} w & =0, \label{CDYBEcompl01}  \\
\partial_{z_0}w  & = w^2 + (\mu + \theta \nu). \label{CDYBEcompl02}
\end{align}
\end{subequations}
As indicated above in (\ref{holo1}), equation (\ref{CDYBEcompl01}) can be understood as a generalized holomorphy condition of the function $w$, i.e. it just depends on the generalized variable $z_0$ but not on $\overline{z}_0$ (precisely like holomorphic functions in complex analysis); and due to this property, as explained in (\ref{holo2}), the equation (\ref{CDYBEcompl02}) is simply
\begin{equation}
 \frac{ d w}{d z_0}  = w^2 + (\mu + \theta \nu)   
\end{equation}
Using the same argument, equations (\ref{DYBEred1}) can be re-written over $\mathbb{R}_\lambda$ as 
\begin{subequations}
\begin{align}
\partial_{\overline{z}_1} w & =0, \label{CDYBEcompl01h1} \\ \qquad \partial_{z_1}w & = -w^2 + (\mu + \theta \nu) \label{CDYBEcompl02h1}
\end{align}
\end{subequations}
but now $w(z_1)=b(z_1)+\theta c(z_1)$ with $z_1=\psi_1 + \theta \gamma_1$ and, again, the holomorphy condition implies  
\begin{equation}
\frac{dw}{dz_1} = -  w^2 + (\mu + \theta \nu).
\end{equation}
Taking into account 
\begin{equation*}
    \mu + \theta \nu = \frac{1}{(\alpha + \theta \beta)^2},
\end{equation*}
then solving the CDYBE for the triples $(\mathfrak{g}_\lambda,\mathfrak{h}_\lambda,K_{\alpha \beta})$ with $\mathfrak{h}_\lambda=\mathfrak{h}_{\lambda,0}$ and $\mathfrak{h}_\lambda=\mathfrak{h}_{\lambda,1}$,  amounts to solve the following two non-linear generalized complex ODEs 
\begin{equation}
    (\alpha + \theta \beta) \frac{d}{dz} ((\alpha + \theta \beta)w) = \pm((\alpha + \theta \beta)w)^2+1.
\end{equation}
For $\mathfrak{h}_{\lambda,0}$ the solution is given by
\begin{equation}\label{w0complsol}
    w(z_0)= \frac{1}{\alpha + \theta \beta} \tan \left( \frac{z_0-C}{\alpha + \theta \beta} \right),
\end{equation}
where
\begin{equation*}
    C= \psi_C + \theta \gamma_C
\end{equation*}
is a constant, while for $\mathfrak{h}_{\lambda,1}$ the solution could be dependent of $z_1$
\begin{equation}\label{w1complsol}
    w(z_1)= \frac{1}{\alpha + \theta \beta} \tanh \left( \frac{z_1-C}{\alpha + \theta \beta} \right)
\end{equation}
with $C \in \mathbb{R}_\lambda$ again a constant, or independent of the complexified variable $z_1$
\begin{equation}\label{constantsol}
    w(z_1)= \pm \frac{1}{\alpha + \theta \beta}= \pm \frac{\alpha - \theta \beta}{\alpha^2 - \lambda \beta^2}.
\end{equation}
Finally, separating (\ref{w0complsol}) into its real-part and $\theta$-part we get $b$ (\ref{bcoeflong}) and $c$ (\ref{ccoeflong}), respectively, with $B$ and $C$ as given in (\ref{coeffbc0}). Proceeding in the same way for (\ref{w1complsol}) or (\ref{constantsol}), we get a similar result but now with B and C as given by (\ref{coeffbc1}) or by (\ref{conscoeff}), respectively. See the appendix D to see in detail how this splitting into real and $\theta$ parts works for every $\lambda \in \mathbb{R}$. 
\end{proof}
In particular, for $\beta=0$, the most general classical dynamical $(\mathfrak{g}_\lambda, \mathfrak{h}_{\lambda,0},K_{\alpha 0})$ $r$-matrix is given by  
\begin{equation}\label{Cathshar0}
r_d^{\beta=0}(\psi_0,\gamma_0)=K_{\alpha 0} + f(\psi_0,\gamma_0)(P^0 \wedge J^0) + b(\psi_0,\gamma_0) (P^1 \wedge J^2 - P^2 \wedge J^1)  + c(\psi_0,\gamma_0)(\lambda J^1 \wedge J^2 + P^1 \wedge P^2)  
\end{equation}
where $f \in C^1(\mathfrak{h}_{\lambda,0}^*)$ is an arbitrary function,
\begin{equation}\label{coeffb0beta0} 
b(\psi_0,\gamma_0)= \frac{1}{\alpha}
\begin{cases}
\frac{\sin \left( \frac{2}{\alpha}(\psi_0-\psi_C )\right)}{\cos \left(\frac{2}{\alpha}(\psi_0-\psi_C)\right) + \cosh \left(\frac{2\sqrt{|\lambda|}}{\alpha} \left(\gamma_0-\gamma_C \right) \right)}, \qquad & \lambda<0 \\
 \tan \left(\frac{\psi_0-\psi_C}{\alpha} \right), \qquad & \lambda=0 \\
\frac{\sin \left( \frac{2}{\alpha}(\psi_0-\psi_C ) \right)}{\cos \left( \frac{2}{\alpha}(\psi_0-\psi_C ) \right)+\cos \left(\frac{2\sqrt{\lambda}}{\alpha}\left(\gamma_0-\gamma_C \right) \right)}, \qquad & \lambda>0
\end{cases}
\end{equation}
and 
\begin{equation}\label{coeffc0beta0}
c(\psi_0,\gamma_0)= \frac{1}{\alpha}
\begin{cases}
\frac{1}{\sqrt{|\lambda|}}\frac{\sinh \left(\frac{2\sqrt{|\lambda|}}{\alpha} \left(\gamma_0-\gamma_C \right) \right)}{\cos \left(\frac{2}{\alpha}(\psi_0-\psi_C) \right) + \cosh \left(\frac{2\sqrt{|\lambda|}}{\alpha} \left(\gamma_0-\gamma_C \right) \right)}, \qquad & \lambda<0 \\
\frac{\gamma_0-\gamma_C}{2 \cos^2 \left(\frac{\psi_0-\psi_C}{\alpha}\right)}, \qquad & \lambda=0 \\
\frac{1}{\sqrt{\lambda}} \frac{\sin \left(\frac{2\sqrt{|\lambda|}}{\alpha} \left(\gamma_0-\gamma_C \right) \right)}{\cos \left( \frac{2}{\alpha}(\psi_0-\psi_C ) \right)+\cos \left(\frac{2\sqrt{\lambda}}{\alpha}\left(\gamma_0-\gamma_C \right) \right)}, \qquad & \lambda>0
\end{cases}.
\end{equation}
Similarly, for $\beta=0$, the most general classical dynamical $(\mathfrak{g}_\lambda, \mathfrak{h}_{\lambda,1},K_{\alpha 0})$ $r$-matrix is given by  
\begin{equation}\label{Cathshar1}
r_d^{\beta=0}(\psi_1,\gamma_1)= K_{\alpha \beta} + f(\psi_1,\gamma_1)(P^1 \wedge J^1)+b(\psi_1,\gamma_1)(P^2 \wedge J^0 - J^2 \wedge P^0) + c(\psi_1,\gamma_1)(\lambda J^2 \wedge J^0 + P^2 \wedge P^0),
\end{equation}
where $f \in C^1(\mathfrak{h}_{\lambda,1}^*)$ is an arbitrary function,
\begin{equation}\label{coeffb1beta0} 
b(\psi_1,\gamma_1)= \frac{1}{\alpha}
\begin{cases}
\frac{\sinh \left( \frac{2}{\alpha}(\psi_1-\psi_C )\right)}{\cosh \left(\frac{2}{\alpha}(\psi_1-\psi_C)\right) + \cos \left(\frac{2\sqrt{|\lambda|}}{\alpha} \left(\gamma_1-\gamma_C \right) \right)}, \qquad & \lambda<0 \\
 \tanh \left(\frac{\psi_1-\psi_C}{\alpha} \right), \qquad & \lambda=0 \\
\frac{\sinh \left( \frac{2}{\alpha}(\psi_1-\psi_C ) \right)}{\cosh \left( \frac{2}{\alpha}(\psi_1-\psi_C ) \right)+\cosh \left(\frac{2\sqrt{\lambda}}{\alpha}\left(\gamma_1-\gamma_C \right) \right)}, \qquad & \lambda>0
\end{cases}
\end{equation}
and 
\begin{equation}\label{coeffc1beta0}
c(\psi_0,\alpha_0)= \frac{1}{\alpha}
\begin{cases}
\frac{1}{\sqrt{|\lambda|}}\frac{\sin \left(\frac{2\sqrt{|\lambda|}}{\alpha} \left(\gamma_1-\gamma_C \right) \right)}{\cosh \left(\frac{2}{\alpha}(\psi_1-\psi_C)\right) + \cos \left(\frac{2\sqrt{|\lambda|}}{\alpha} \left(\gamma_1-\gamma_C \right) \right)}, \qquad & \lambda<0 \\
\frac{\gamma_1-\gamma_C}{2 \cosh^2 \left(\frac{\psi_1-\psi_C}{\alpha}\right)}, \qquad & \lambda=0 \\
\frac{1}{\sqrt{\lambda}} \frac{\sinh \left(\frac{2\sqrt{|\lambda|}}{\alpha} \left(\gamma_1-\gamma_C \right) \right)}{\cosh \left( \frac{2}{\alpha}(\psi_1-\psi_C ) \right)+\cosh \left(\frac{2\sqrt{\lambda}}{\alpha}\left(\gamma_1-\gamma_C \right) \right)}, \qquad & \lambda>0
\end{cases}.
\end{equation}
These correspond to extensions, for both signatures and any value of the cosmological constant $\Lambda_C$, of the classical dynamical $r$-matrices found in \cite{Meusburger:2012wc} (Lemma 4.11) for $\Lambda_C=0$ in the Lorentzian case (with $\psi_C=\gamma_C=0$, $\alpha=2$ and $f \equiv 0$). \\\\
Theorem \ref{MainTheo2} provides all the elements in $\text{Dyn}(\mathfrak{g}_\lambda,\mathfrak{h}_\lambda,K_{\alpha \beta})$ for $\mathfrak{h}_\lambda=\mathfrak{h}_{\lambda,0}$ and $\mathfrak{h}_\lambda=\mathfrak{h}_{\lambda,1}$, parametrized in terms of $(\psi_C,\gamma_C) \in \mathbb{R}^2$ and $f \in C^1(\mathfrak{h}_\lambda^*)$ in both cases. In order to achieve a full description of the moduli space of classical dynamical $r$-matrices $\mathcal{M}^{\mathfrak{C}}(\mathfrak{g},K_{\alpha, \beta})$ we first need to consider the quotients  
\begin{equation*}
\text{Dyn}(\mathfrak{g}_\lambda,\mathfrak{h}_\lambda,K_{\alpha \beta})/\mathcal{G}(\mathfrak{g}_\lambda,\mathfrak{h}_\lambda)  
\end{equation*}
for both $\mathfrak{h}_{\lambda,0}$ and $\mathfrak{h}_{\lambda,1}$ (if it is the case).
\begin{lemma}\label{Lemma3}
The space of orbits of classical dynamical ($\mathfrak{g}_\lambda$,$\mathfrak{h}_\lambda$,$K_{\alpha \beta}$) $r$-matrices with respect to the action of dynamical ($\mathfrak{g}_\lambda$,$\mathfrak{h}_\lambda$) gauge transformations is parametrized by pairs $(\psi_C,\gamma_C) \in \mathbb{R}^2$, for both $\mathfrak{h}_{\lambda}=\mathfrak{h}_{\lambda,0}$ and $\mathfrak{h}_{\lambda}=\mathfrak{h}_{\lambda,1}$. 
\end{lemma}
\begin{proof} 
Theorem \ref{MainTheo2} states that any classical dynamical ($\mathfrak{g}_\lambda$,$\mathfrak{h}_{\lambda,i}$,$K_{\alpha \beta}$) $r$-matrix are of the form  
\begin{equation*}
r_d(\psi_0,\gamma_0)=K_{\alpha \beta} +  f(\psi_0,\gamma_0)(P^0 \wedge J^0)+b(\psi_0,\gamma_0) (P^1 \wedge J^2 - P^2 \wedge J^1) + c(\psi_0,\gamma_0)(\lambda J^1 \wedge J^2 + P^1 \wedge P^2) 
\end{equation*}
and 
\begin{equation*}
r_d(\psi_1,\gamma_1)= K_{\alpha \beta} + f(\psi_1,\gamma_1)(P^1 \wedge J^1)+b(\psi_1,\gamma_1)(P^2 \wedge J^0 - J^2 \wedge P^0) + c(\psi_1,\gamma_1)(\lambda J^2 \wedge J^0 + P^2 \wedge P^0),
\end{equation*}
for $i=0$ and $i=1$, respectively, with $f: \in \mathcal{F}(\mathfrak{h}_{\lambda,i}^*)$ an arbitrary function and $b,c: \mathfrak{h}_{\lambda,i}^* \to \mathbb{R}$ given by (\ref{bcoeflong}) and (\ref{ccoeflong}).\\\\ 
All the elements of $\mathcal{G}(\mathfrak{g}_\lambda,\mathfrak{h}_{\lambda,i})$ are of the form    
\begin{equation*}
g(\psi_i,\gamma_i)=\exp(m(\psi_i,\gamma_i)J_i+s(\psi_i,\gamma_i)P_i)
\end{equation*}
with $m,s: \mathfrak{h}_{\lambda,i}^* \to \mathbb{R}$ smooth maps. \\\\ 
Hence the action of any $g \in \mathcal{G}(\mathfrak{g}_\lambda, \mathfrak{h}_{\lambda,i})$ over any element $r_d \in \text{Dyn}(\mathfrak{g}_\lambda,\mathfrak{h}_{\lambda,i},K_{\alpha \beta})$ is explicitly given by 
\begin{equation*}
\begin{split}
  g \rhd r_d  & = (\text{Ad}(g) \otimes \text{Ad}(g))\left[r_{d}- \left( \frac{\partial s}{\partial \gamma_i} - \frac{\partial m}{\partial \psi_i}   \right) P^i \wedge J^i \right] \\
  & = r_{d}- \left( \frac{\partial s}{\partial \gamma_i} - \frac{\partial m}{\partial \psi_i}   \right) P^i \wedge J^i. 
  \end{split}
\end{equation*}
Since, by the  Poincar\'{e} Lemma, for every function $f \in \mathcal{F}(\mathfrak{h}_{\lambda,i}^*)$ there exist smooth functions $s,m:\mathfrak{h}_{\lambda,i}^* \to \mathbb{R}$ such that 
\begin{equation*}
 \frac{\partial s}{\partial \gamma_i}-\frac{\partial m}{\partial \psi_i} = f, 
\end{equation*}
we conclude that any classical dynamical ($\mathfrak{g}_\lambda$,$\mathfrak{h}_{\lambda,i}$,$K_{\alpha \beta}$) $r$-matrix is gauge equivalent to
\begin{equation}\label{gaugeequiv0}
r_d(\psi_0,\gamma_0)=K_{\alpha \beta} +b(\psi_0,\gamma_0) (P^1 \wedge J^2 - P^2 \wedge J^1) + c(\psi_0,\gamma_0)(\lambda J^1 \wedge J^2 + P^1 \wedge P^2) 
\end{equation}
and 
\begin{equation}\label{gaugeequiv1}
r_d(\psi_1,\gamma_1)= K_{\alpha \beta} +b(\psi_1,\gamma_1)(P^2 \wedge J^0 - J^2 \wedge P^0) + c(\psi_1,\gamma_1)(\lambda J^2 \wedge J^0 + P^2 \wedge P^0),
\end{equation}
for $i=0$ and $i=1$, respectively, with $b$ and $c$ as given in Theorem 2. \\\\
Hence, since the functions $b$ and $c$ above are uniquely determined by the constants $\psi_C$ and $\gamma_C$, it follows each of the quotients  $\mathcal{M}^\mathfrak{C}(\mathfrak{g}_\lambda,\mathfrak{h}_{\lambda,i},K_{\alpha \beta})$ is in bijection with $\mathbb{R}^2$. 
\end{proof}
Theorem \ref{MainTheo2} and Lemma \ref{Lemma3} provides a parametrization and complete description of all the elements in the moduli space of Cartan classical dynamical $r$-matrices associated to $(\mathfrak{g}_\lambda,K_{\alpha \beta})$
\begin{equation*}
  \mathcal{M}^{\mathfrak{C}}(\mathfrak{g}_\lambda,K_{\alpha \beta})    =  \mathcal{M}^\mathfrak{C}(\mathfrak{g}_\lambda,\mathfrak{h}_{\lambda,0},K_{\alpha \beta}) \sqcup \mathcal{M}^\mathfrak{C}(\mathfrak{g}_\lambda,\mathfrak{h}_{\lambda,1},K_{\alpha \beta})
\end{equation*}
for every $\lambda, \alpha, \beta \in \mathbb{R}$. This amounts, following the explanation in Subsection 1.4, to a full description of all the Poisson structures over the gauge-fixed space of $G_\lambda$-flat connections over Riemann surfaces. \\\\
Finally, the set of \textit{all} the Cartan classical dynamical $r$-matrices associated to the pair $(\mathfrak{g}_\lambda,K_{\alpha \beta})$ is then  given by   
\begin{equation*}
    \text{Dyn}^{\mathfrak{C}}(\mathfrak{g}_\lambda,K_{\alpha \beta})= \left[ \bigcup_{g \in G}\text{Dyn}^g(\mathfrak{g}_\lambda,\mathfrak{h}_{\lambda,0},K_{\alpha \beta}) \right] \sqcup \left[ \bigcup_{g \in G}\text{Dyn}^g(\mathfrak{g}_\lambda,\mathfrak{h}_{\lambda,1},K_{\alpha \beta}) \right]. 
\end{equation*}

\subsection{Dynamical generalizations of classical \texorpdfstring{$r$}{}-matrices for \texorpdfstring{$\mathfrak{g}_\lambda$}{}}
A systematic algebraic analysis of the CYBE for $\mathfrak{g}_\lambda$ and the derivation of some particular solutions were considered in \cite{Osei:2017ybk}. In this spirit now we examine some particular solutions of the CDYBE for $\mathfrak{g}_\lambda$, using for this purpose the equivalent description derived in Theorem \ref{MainTheo1}. We recall that even though we are mainly interested in classical dynamical $r$-matrices (where the equivariance condition is required), since these appear in the Poisson structure of gauge-fixed character varieties, the space of solutions of the CDYBE has been studied too (as mentioned  in passing  in our Introduction) and also, it is  interesting to see what dynamical generalizations of some well-known solutions of the CYBE for $\mathfrak{g}_\lambda$ look like. \\\\ 
Following the presentation in \cite{Osei:2017ybk}, we consider solutions of the CDBYE (\ref{CDYB}) associated to $(\mathfrak{g}_\lambda,\mathfrak{h}_{\lambda,0},K_{\alpha \beta})$ 
\begin{equation*}
r(\psi_0,\gamma_0)  = K_{\alpha \beta} + A_{ab}(\psi_0,\gamma_0)J^a \otimes J^b + B_{ab}(\psi_0,\gamma_0)(P^a \otimes J^b - J^b \otimes P^a) + C_{ab}(\psi_0,\gamma_0)P^a \otimes P^b
\end{equation*}
for the cases where (I) $B$ is diagonal and $C \equiv 0$, (II) $A=\lambda C$ and $B$ is skew-symmetric and (III) $A=C \equiv 0$ and B is skew-symmetric, obtaining in this way dynamical generalizations of the so-called (see Table 1 in \cite{Osei:2017ybk} for the terminology) \textit{classical doubles}, \textit{generalised complexifications} and \textit{kappa-Poincar\'e} $r$-matrices, respectively.\\\\
Here we focus on the $\mathfrak{h}_{\lambda}=\mathfrak{h}_{\lambda,0}$ case, since for $\mathfrak{h}_{\lambda,1}$ analogous solutions of the CDYBE are obtained (as we observed in the previous subsection) and then via the $G_\lambda$-action (\ref{mapdyn}) all the solutions presented in this section can be derived for any Cartan subalgebra of $\mathfrak{g}_\lambda$.
\textbf{(I) Solutions with $B$ diagonal and trivial $C$.} If we consider solutions with $C \equiv 0$, the equation (\ref{main2}) forces $\nu=0$, while the other 3 equations (\ref{main1}), (\ref{main3}) and (\ref{main4}) reduce to \begin{equation}\label{Czero}
\begin{split}
\frac{1}{2}\text{tr}(A^2)-\frac{\lambda}{2}\left[ \text{tr}(B)^2 - \text{tr}(B^2) \right]+\text{div}_{\bm{\alpha}} (v^A) & = \mu \lambda \\
[B-\text{tr}(B)\text{id}](B+B^t) + \frac{1}{2}[\text{tr}(B)^2-\text{tr}(B^2)]\text{id} + \text{curl}_{\bm{\psi}}(B) & = - \mu \text{id} \\
-A(B+B^t)- (B^t-\text{tr}(B)\text{id})A - \text{tr}(AB)\text{id} + \text{curl}_{\bm{\alpha}} (B^t) - \text{grad}_{\bm{\psi}}(v^A)& = 0 
\end{split}
\end{equation}
Denoting $v^{A}(\psi_0,\gamma_0)$ by $(a_0(\psi_0,\gamma_0),a_1(\psi_0,\gamma_0),a_2(\psi_0,\gamma_0))$ and taking 
\begin{equation*}
  B(\alpha_0,\psi_0)=\text{diag}(b_0(\psi_0,\gamma_0),b_1(\psi_0,\gamma_0),b_2(\psi_0,\gamma_0)), 
\end{equation*}
(\ref{Czero}) leads to the following system of PDEs    
\begin{subequations}
    \begin{align}
- \sum_{i=0}^2 a_i^2 - \lambda (b_0b_1+b_0b_2+b_1b_2) +  \frac{\partial a_0}{\partial \gamma_0} & = \mu \lambda, \label{doub1} \\
\begin{pmatrix} b_1b_2-b_0(b_1+b_2) & 0 & 0 \\ 0 & b_0b_2-b_1(b_0+b_2) & \frac{\partial b_1}{\partial \psi_0} \\ 0 & -\frac{\partial b_2}{\partial \psi_0} & b_0b_1-b_2(b_0+b_1) \end{pmatrix} & = -\mu \text{id}, \label{doub2} \\  
\begin{pmatrix} -\frac{\partial a_0}{\partial \psi_0} & a_2(b_1-b_2)  -\frac{\partial a_1}{\partial \psi_0} & -a_1(b_1-b_2) - \frac{\partial a_2}{\partial \psi_0}\\ -a_2(b_0-b_2)  & 0 & a_0(b_0-b_2) + \frac{\partial b_1}{\partial \gamma_0} \\ a_1(b_0-b_1) & -a_0(b_0-b_1)-\frac{\partial b_2}{\partial \gamma_0} & 0
\end{pmatrix} & = 0. \label{doub3}
\end{align}
\end{subequations}
From equations (\ref{doub2}) and (\ref{doub3}) we conclude
\begin{equation}
B= \pm \sqrt{\mu} \text{id}
\end{equation}
and 
\begin{equation}
A(\psi_0,\gamma_0)=A(\gamma_0),
\end{equation}
while equation (\ref{doub1}) reduces to a constraint over the vector $v^A$ associated to the anti-symmetric matrix $A$, 
\begin{equation}
\langle v^A(\gamma_0) , v^A(\gamma_0) \rangle - \frac{d a_0}{d \gamma_0}= -4 \mu \lambda.
\end{equation}
Hence for diagonal $B$ and trivial $C$, we have found there exist solutions to the CDYBE for $(\mathfrak{g}_\lambda,\mathfrak{h}_{\lambda,0},K_{\alpha \beta})$ of the form  
\begin{equation}\label{dyndoub1}
\begin{split}
r^{\beta=0}_D(\gamma_0) & =\frac{1}{\alpha}(J_a \otimes P^a + P_a \otimes J^a) + \epsilon^{abc} v_a(\gamma_0) J_b \otimes J_c \pm \sqrt{\mu} (P^a \otimes J_a - J_a \otimes P^a) \\
& = \left(\frac{1}{\alpha} \mp \sqrt{\mu} \right)J_a \otimes P^a + \left(\frac{1}{\alpha} \pm \sqrt{\mu} \right)P_a \otimes J^a + \epsilon^{abc}v_a(\gamma_0) J_b \otimes J_c
\end{split}
\end{equation}
if $\beta=0$, or 
\begin{equation}\label{dyndoub2}
\begin{split}
r^{\alpha=0}_D(\gamma_0) & =\frac{1}{\lambda \beta} (\lambda J_a \otimes J^a + P_a \otimes P^a) + \epsilon^{abc} v_a(\gamma_0) J_b \otimes J_c \pm \sqrt{\mu} (P^a \otimes J_a - J_a \otimes P^a) 
\end{split}
\end{equation}
if $\alpha=0$, where the choice of sing depends if we are dealing with the Euclidean or Lorentzian case, respectively. \\\\
The family of classical dynamical $r$-matrices (\ref{dyndoub1}) is a dynamical generalization of the well-known classical $r$-matrices associated to the standard double bialgebra structures over $\mathfrak{g}_\lambda$. Similarly, the family (\ref{dyndoub2}) is the dynamical generalization of the classical $r$-matrices related to \textit{exotic} double structures over $\mathfrak{g}_\lambda$. 
    

 
\textbf{(II) Solutions with $A=\lambda C$ and $B$ skew-symmetric.}  If now we demand for solutions of the CDYBE such that $A= \lambda C$, the set (\ref{main1})-(\ref{main4}) reduces to the following system of coupled PDEs
\begin{equation}\label{alambdac}
\begin{split}
\frac{\lambda}{2}\text{tr}(C^2)-\frac{1}{2}[\text{tr}(B)^2-\text{tr}(B^2)]+\text{div}_{\bm{\gamma}}(v^C) & =\mu, \\
\text{tr}(BC)+\text{div}_{\bm{\psi}}(v^	C) & =\nu, \\
[B-\text{tr}(B)\text{id}](B+B^t)+\text{div}_{\bm{\gamma}}(v^C)\text{id} + \text{curl}_{\bm{\psi}}(B) - \text{grad}_{\bm{\gamma}}(v^C) & = 0, 	\\
-\lambda C (B+B^t) + \lambda[ \text{div}_{\bm{\psi}}(v^C)\text{id}-\text{grad}_{\bm{\psi}}(v^C)] + \text{curl}_{\bm{\gamma}}(B^t) & = 0.
\end{split}
\end{equation}
This system is still complicated enough to solve in a general way, reason why we focused on some specific solutions corresponding to a particular choice of the maps $C$ and $B$: For $B=[v^B(\psi_0,\gamma_0), \ \cdot \ ]$ and $C=[v^C(\psi_0,\gamma_0), \ \cdot \ ]$, with $v^B(\psi_0,\gamma_0)$ and $ v^C(\psi_0,\gamma_0)$ in $C^1(\mathfrak{h}_{\lambda,0}^*) \otimes \mathbb{R}^3$ or in $C^1(\mathfrak{h}_{\lambda,0}^*) \otimes \mathbb{M}^{1,2}$, given by 
\begin{equation*}
v^B(\psi_0,\gamma_0)=(b_0(\psi_0,\gamma_0),b_1(\psi_0,\gamma_0),b_2(\psi_0,\gamma_0) \text{ and } v^C(\psi_0,\gamma_0)=(c_0(\psi_0,\gamma_0),c_1(\psi_0,\gamma_0),c_2(\psi_0,\gamma_0)), 
\end{equation*}
the set of equations (\ref{alambdac}) reduces to 
\begin{equation}\label{CDYBECath}
\begin{split}
-\lambda \langle v^C,v^C \rangle - \langle v^B,v^B \rangle+  \frac{\partial c_0}{\partial \gamma_0} & =\mu, \\
-2 \langle v^B,v^C \rangle +  \frac{\partial c_0}{\partial \psi_0} & = \nu, \\
\frac{\partial b_i}{\partial \psi_0} - \frac{\partial c_i}{\partial \gamma_0} & = 0 \text{ for } i=0,1,2, \\
\lambda \frac{\partial c_i}{\partial \psi_0} - \frac{\partial b_i}{\partial \gamma_0} & = 0 \text{ for } i=0,1,2   
\end{split}
\end{equation}
which precisely, reduces to the system (\ref{CYBEFIN}) associated to classical dynamical $r$-matrices in the case $v^B \parallel \bm{e}_0$ and $v^C \parallel \bm{e}_0$.   \\\\
\textbf{(III) Solutions with $A$ and $C$ trivial and $B$ skew-symmetric.} If in the previous case, we additionally consider $C=0$, then (\ref{CDYBECath}) reduces to
\begin{equation}
\begin{split}
\langle v^B,v^B \rangle & = - \mu \\
\frac{\partial b_0}{\partial \gamma_0}=\frac{\partial b_1}{\partial \gamma_0}=\frac{\partial b_2}{\partial \gamma_0} & =0 \\
\frac{\partial b_0}{\partial \psi_0}=\frac{\partial b_1}{\partial \psi_0}=\frac{\partial b_2}{\partial \psi_0} & =0 \end{split}
\end{equation}
In other words, $v^B$ is a \textit{constant} vector function (just an element in $\mathbb{R}^3$ or $\mathbb{M}^{1,2}$) whose (pseudo)norm is given by $-\mu$. Therefore, we find there are not dynamical generalizations of the classical \textit{Kappa-Poincar\'e} $r$-matrices. Explicitly, since $\nu=0$, there exist just two possible \textit{dynamical} $r$-matrices of this type given by 
\begin{equation}\label{kappa1}
r_{\kappa P}^{\beta=0} = \frac{1}{\alpha}(J_a \otimes P^a + P_a \otimes J^a) + \epsilon_{abc}v^a (P^b \otimes J^c -  J^b \otimes P^c) \quad \text{with} \quad \langle v,v \rangle = - \frac{1}{\alpha^2}
\end{equation}
or
\begin{equation}\label{kappa2}
r_{\kappa P}^{\alpha =0} = \frac{1}{\lambda \beta}(\lambda J_a \otimes J^a + P_a \otimes P^a) + \epsilon_{abc}v^a (P^b \otimes J^c -  J^b \otimes P^c) \quad \text{with} \quad \langle v,v \rangle = - \frac{1}{\lambda \beta^2}
\end{equation}
which are precisely the 2 possible cases described in \cite{Meusburger:2008rt} (equations (5.5) and (5.6)). \\\\
Therefore $r_{\kappa P}^{\beta=0}$ and $r_{\kappa P}^{\alpha=0}$ correspond to classical dynamical $(\mathfrak{g}_\lambda,\mathfrak{h}_{\lambda,0},K_{\alpha 0})$ and $(\mathfrak{g}_\lambda,\mathfrak{h}_{\lambda,0},K_{0 \beta})$ r-matrices, respectively, if and only if $v \in \text{Span}\{\bm{e}_0\}$. In other words, the $\mathfrak{h}_{\lambda,0}$ equivariance condition constrains the vector defining the classical dynamical $r$-matrix to have the same direction of the generators ($J_0$ and $P_0$) of $\mathfrak{h}_{\lambda,0}$.

\section{The FGMPP dynamical \texorpdfstring{$r$}{}-matrices for \texorpdfstring{$\mathfrak{g}_\lambda$}{}}\label{Sec4paper1}
Theorem \ref{MainTheo2} shows that \textit{dynamical generalized complexifications}, i.e. with coefficients solving (\ref{CDYBECath}), are the only solutions of the CDYBE that actually are classical dynamical $r$-matrices and therefore have a natural origin in the setting of gauge-fixed character varieties (see e.g. \cite{Meusburger:2011oew} and \cite{Meusburger:2012wc}). In this section we show how these solutions are related (indeed coincide up to dynamical gauge transformations) with a family of classical dynamical $r$-matrices found by Feher, Gabor, Marshall, Palla and Pusztai when studing the (quasi-)Poisson structures of the chiral WZNW model (see \cite{Feher:2001th}). \\\\
Alekseev and Meinrenken (AM) \cite{Alekseev1999TheNW} found that for any \textit{self-dual} Lie algebra $\mathfrak{g}$, i.e. equipped with a non-degenerate $Ad$-invariant symmetric bilinear form $\langle \cdot, \cdot \rangle$, there exists a \textit{canonical} classical dynamical $(\mathfrak{g},\mathfrak{g},K)$ $r$-matrix $r_{\text{AM}}: \mathfrak{g}^* \to \mathfrak{g} \otimes \mathfrak{g}$, where $K$ is the element in $S^2(\mathfrak{g})^{\mathfrak{g}}$ associated to $\langle \cdot, \cdot \rangle$ and $\mathfrak{g}^*$ is the dual of $\mathfrak{g}$ with respect to the bilinear form, given by 
\begin{equation}\label{AMFull}
r_{\text{AM}}(x)= K + \rho_{\text{AM}}(x) 
\end{equation}
such that
\begin{equation}\label{AMF}
\begin{split}
\rho_{\text{AM}}: \mathfrak{g}^* & \to \mathfrak{g} \wedge \mathfrak{g} \\
x & \mapsto f(1 \otimes \text{ad}[K^{\lor}(x)])(K)
\end{split}
\end{equation}
with $K^{\lor}$ the linear map $\mathfrak{g}^* \to \mathfrak{g}$ associated to $K$ and 
\begin{equation}
f(z)=\coth(z) - \frac{1}{z}.
\end{equation}
The adjective \textit{canonical} was introduced by \cite{Etingof2000OnTM} and extensively used since then (see e.g. \cite{pusztai2001note}), given that any classical dynamical $(\mathfrak{g},\mathfrak{g},K)$ is dynamical gauge equivalent to $r_{\text{AM}}(x-x_0)$ for a shift $x_0 \in \mathfrak{g}^*$.  \\\\
Analogously, as explained for gauge-fixed character varieties (\ref{gfixbra}), this classical dynamical ($\mathfrak{g},\mathfrak{g},K$) $r$-matrix defines a Poisson structure over 
\begin{equation*}
    \mathfrak{g}^* \times G
\end{equation*}
given by   
\begin{subequations}\label{PoissF}
    \begin{align}
\{t^a,t^b\} & = -f^{ab}_{c}t^c, \\
\{F,t^a\} & = (R_{t_a}+L_{t_a})F, \\
\{F, \tilde{F} \}(t)& = \Pi_{\text{ext}}^{0,1}(r_{\text{AM}}(t))(dF,d\tilde{F}), 
    \end{align}
\end{subequations}
where $F,\tilde{F} \in C^\infty(G)$, $\{t^a\}_{a=1,\cdots,\text{dim}\mathfrak{g}}$ is a basis of $\mathfrak{g}$ and $\{t_a\}_{a=1,\cdots,\dim \mathfrak{g}}$ its dual in $\mathfrak{g}^*$. \\\\
Then we consider a self-dual Lie subalgebra of $\mathfrak{h}$ of $\mathfrak{g}$ (i.e. such that the restriction $\langle \cdot, \cdot \rangle \big|_{\mathfrak{h}}$ is non-degenerate) and decompose (even adapting) the basis $\{t^a\}_{a=1,\cdots,\dim \mathfrak{g}}$ of $\mathfrak{g}$ into a basis $\{h^a\}_{a=1,\cdots, \dim \mathfrak{h}}$ for $\mathfrak{h}$ and $\{\tilde{h}^a\}_{a=1,\cdots, \dim \mathfrak{h}^\perp}$ for $\mathfrak{h}^\perp$, since $\mathfrak{g}=\mathfrak{h} + \mathfrak{h}^{\perp}$. Feher, Gabor, Marshall, Palla and Pusztai in a series of papers (see (\cite{BALOG1990227}, \cite{balog1999chiral}, \cite{Fehr2001OnDR}, \cite{Feher:2001th}, \cite{Feher:2001uu}, \cite{feher2004non}) considered the reduced subspace of $\mathfrak{g}^* \times G$ obtained by imposing the first class constraints $\tilde{h}^a \approx 0$ for $a=1,\cdots,\dim \mathfrak{h}^\perp$. In other words, by seeing $\{\tilde{h}^a\}_{a=1,\cdots, \dim \mathfrak{h}^\perp}$ as functions over $\mathfrak{g}^* \times G$, the Poisson structure (\ref{PoissF}) is modified in such a way these functions become Poisson functions with respect to the new (Dirac) bracket. \\\\
Exactly as before, these new Poisson structure is obtained via gauge fixing, getting so a Poisson structure  
\begin{equation*}
\{F,\tilde{F}\}_{D}=\{F,\tilde{F}\}-\{F_1,e^a\}C_{ab}^{-1}(t_{\mathfrak{h}})\{e^b,\tilde{F}\}
\end{equation*}
over 
\begin{equation*}
    \mathfrak{h}^* \times G
\end{equation*}
(after strongly imposing the constraints $\tilde{h}^a=0$ for $a=1,\cdots,\dim \mathfrak{h}^\perp$) where
\begin{equation*}
C^{ab}(t_{\mathfrak{h}})=\{t^a,t^b\}|_{t_{\mathfrak{h}^{\perp}}=0} = C(t_{\mathfrak{h}}) = \text{ad} (t_{\mathfrak{h}})|_{\mathfrak{h}^\perp}.
\end{equation*}
Thus, this gauge fixing amounts to modify (\ref{PoissF}) by changing the skew-symmetric part of the AM classical dynamical $r$-matrix by   
\begin{equation}
\rho_{\text{FGMPP}}(x) = \rho_{\text{AM}}(x) + \frac{1}{\text{ad}(x)}, 
\end{equation}
getting in this way a Poisson structure over $\mathfrak{h}^* \times G$ with Dirac brackets given by
\begin{subequations}
        \begin{align}
\{h^a,h^b\}_D & = 0, \\
\{F,h^a\}_D & = (R_{h^a}+L_{h^a})F, \\
\{F, \tilde{F} \}_D(h)& = \Pi_{ext}^{0,1}(r_{\text{FGMPP}}(h))(dF,d\tilde{F}),
    \end{align}
\end{subequations}
where the \textit{canonical} classical dynamical ($\mathfrak{g},\mathfrak{h},\Omega$) $r$-matrix (denoted here by FGMPP) is given by 
\begin{equation}\label{FGMPP}
r_{\text{FGMPP}}(x)= K + \rho_{\text{FGMPP}}(x) 
\end{equation}
such that
\begin{equation}\label{FGP}
\begin{split}
\rho_{\text{FGMPP}}: \mathfrak{h}^* & \to \mathfrak{h}^{\perp} \wedge \mathfrak{h}^{\perp} \\
x & \mapsto g(1 \otimes \text{ad}[K^{\lor}(x)])(K)
\end{split}
\end{equation}
with
\begin{equation}\label{gfixg}
g(z)= \coth(z).
\end{equation}
In this section we compute the FGMPP classical dynamical $r$-matrices associated to  $(\mathfrak{g}_\lambda,\mathfrak{h}_{\lambda,i},K_{\alpha \beta})$ for $i=0,1$ and any $\alpha, \beta \in \mathbb{R}$, proving by direct computation the following. 
\begin{theorem}\label{MainTheo3}
The \textit{generalized complexified} classical dynamical ($\mathfrak{g}_\lambda,\mathfrak{h}_{\lambda,i}, K_{\alpha \beta}$) $r$-matrices (\ref{gaugeequiv0}) and (\ref{gaugeequiv1}) are gauge equivalent to the FGMPP classical dynamical $r$-matrices (\ref{FGMPP}) associated to ($\mathfrak{g}_\lambda,\mathfrak{h}_{\lambda,i}, K_{\alpha \beta}$).    
\end{theorem}
\begin{proof}
We start by considering the $i=0$ case. Since
\begin{equation*}
K_{\alpha \beta} \equiv \frac{\alpha}{\alpha^2-\lambda \beta^2}(J_a \otimes P^a + P_a \otimes J^a) - \frac{\beta}{\alpha^2-\lambda \beta^2} (\lambda J_a \otimes J^a + P_a \otimes P^a),
\end{equation*}
for a general element $x_0 =\gamma_0 J_0^* + \psi_0 P_0^*$ in $\mathfrak{h}^*_{\lambda,0}$ we have 
\begin{equation*}
    K^\lor_{\alpha \beta}(x_0)= \xi(\gamma_0,\psi_0,\alpha,\beta) J_0 + \zeta(\gamma_0,\psi_0,\alpha,\beta)P_0
\end{equation*}
with 
\begin{equation*}
\xi(\gamma_0,\psi_0,\alpha,\beta)= \psi_0 \left( \frac{ \alpha}{\alpha^2-\lambda \beta^2} \right) - \gamma_0 \left( \frac{ \beta \lambda}{\alpha^2-\lambda \beta^2} \right)
\end{equation*}
and 
\begin{equation*}
\zeta(\gamma_0,\psi_0,\alpha,\beta)= \gamma_0 \left( \frac{ \alpha}{\alpha^2-\lambda \beta^2} \right) - \psi_0 \left( \frac{ \beta }{\alpha^2-\lambda \beta^2} \right).
\end{equation*}
Note these are precisely the arguments of the functions $B$ and $C$ in (\ref{bcoeflong}) and (\ref{ccoeflong}). \\\\
Thus, according to (\ref{commut}), we get  
\begin{equation*}
\text{ad}[K_{\alpha \beta}^{\lor}(x_0)]|_{\mathfrak{h}^\perp_{\lambda,0}}: \begin{pmatrix}  J_1 \\ J_2 \\ P_1 \\ P_2 \end{pmatrix} \mapsto \xi \begin{pmatrix}  -J_2 \\ J_1 \\  -P_2 \\ P_1 \end{pmatrix} + \zeta \begin{pmatrix}  -P_2 \\ P_1  \\ -\lambda J_2 \\ \lambda J_1 \end{pmatrix} = \begin{pmatrix} -( \xi J_2 + \zeta P_2)  \\ \xi J_1 +\zeta P_1 \\ -(\xi P_2 + 	\zeta \lambda J_2) \\ \xi P_1 + \zeta \lambda J_1   \end{pmatrix},
\end{equation*}
where we have omitted the dependence of $\xi$ and $\zeta$ on $\gamma_0,\psi_0,\alpha$ and $\beta$. Hence, the matrix representation of $\text{ad}[K_{\alpha \beta}^{\lor}(x_0)]|_{\mathfrak{h}^\perp_{\lambda,0}}$ in the (standard) basis $\{J_1,J_2,P_1,P_2\}$ of $\mathfrak{h}_{\lambda,0}^{\perp}$ is given by 
\begin{equation}\label{admatrix}
\text{ad}[K_{\alpha \beta}^{\lor}(x_0)]|_{\mathfrak{h}^\perp_{\lambda,0}}= \begin{pmatrix} 0 &  \xi & 0 & \zeta \lambda \\
-\xi & 0 & -\zeta\lambda & 0 \\ 0 & \zeta  & 0 & \xi \\ -\zeta  & 0 & -\xi & 0 \end{pmatrix}.
\end{equation}
As indicated in Appendix C, by the Weierstrass factorization theorem, the function $g(z)$ is equal to
\begin{equation}\label{expancoth}
g(z)= \frac{1}{z} + \sum_{n=1}^\infty  \frac{2z}{(\pi n)^2+z^2}.
\end{equation}
Therefore, given that 
\begin{equation*}
\frac{\text{ad}[K_{\alpha \beta}^{\lor}(x_0)]|_{\mathfrak{h}^\perp_{\lambda,0}}}{( \pi n)^2+ \left(\text{ad}[K_{\alpha \beta}^{\lor}(x_0)]|_{\mathfrak{h}^\perp_{\lambda,0}} \right)^2}=\begin{pmatrix} 0 & F_n(\zeta,\xi,\lambda) & 0  & \lambda G_n(\zeta,\xi,\lambda) \\ -F_n(\zeta,\xi,\lambda) & 0 &  -\lambda G_n(\zeta,\xi,\lambda) & 0 \\ 0 &  G_n(\zeta,\xi,\lambda) & 0  & F_n(\zeta,\xi,\lambda) \\   -G_n(\zeta,\xi,\lambda) & 0 & -F_n(\zeta,\xi,\lambda) & 0 \end{pmatrix}
\end{equation*}
where
\begin{equation*}
F_n(\zeta,\xi,\lambda)  =\frac{\xi[( \pi n)^2-(\xi^2-\lambda \zeta^2)]}{[( \pi n)^2-(\xi^2+\lambda \zeta^2)]^2-(2 \zeta \xi)^2 \lambda}, \quad 
G_n(\alpha,\psi,\lambda)  =\frac{\zeta[( \pi n)^2+(\xi^2-\lambda \zeta^2)]}{[(\pi n)^2-(\xi^2+\lambda \zeta^2)]^2-(2 \zeta \xi)^2 \lambda}     
\end{equation*}
and the matrix representation of the inverse of $\text{ad}[K_{\alpha \beta}^{\lor}(x_0)]|_{\mathfrak{h}^\perp_{\lambda,0}}$ is  \begin{equation*}
\frac{1}{\text{ad}[K_{\alpha \beta}^{\lor}(x_0)]|_{\mathfrak{h}^\perp_{\lambda,0}}}= \frac{1}{\xi^2- \lambda \zeta^2} \begin{pmatrix} 0 & -\xi & 0 & \lambda \zeta \\  \xi & 0 & - \lambda \zeta & 0 \\ 0 & \zeta & 0 & -\xi \\ -\zeta  & 0 & \xi & 0 \end{pmatrix},
\end{equation*}
then we obtain the matrix representation of $g(\text{ad}[K_{\alpha \beta}^{\lor}(x_0)]|_{\mathfrak{h}^\perp_{\lambda,0}})$ is 
\begin{equation}\label{mat}
g \left(\text{ad}[K_{\alpha \beta}^{\lor}(x_0)]|_{\mathfrak{h}^\perp_{\lambda,0}}\right)= \begin{pmatrix}  0 & F(\zeta,\xi,\lambda) & 0 &  \lambda G(\zeta,\xi,\lambda) \\ -F(\zeta,\xi,\lambda) & 0 &  -\lambda G(\zeta,\xi,\lambda) & 0 \\ 0 &  G(\zeta,\xi,\lambda) & 0 & F(\zeta,\xi,\lambda) \\ -G(\zeta,\xi,\lambda) & 0 & -F(\zeta,\xi,\lambda) & 0  \end{pmatrix}
\end{equation}
with 
\begin{equation*}
\begin{split}
F(\zeta,\xi,\lambda) & = -\frac{\xi}{\xi^2-\lambda \zeta^2} + 2\sum_{n=1}^{\infty}F_n(\zeta,\xi,\lambda)\\
& = -\frac{\xi}{\xi^2-\lambda \zeta^2} + \sum_{n=1}^{\infty} \frac{2\xi[(\pi n)^2-(\xi^2-\lambda \zeta^2)]}{[(\pi n)^2-(\xi^2+\lambda \zeta^2)]^2-(2 \zeta \xi)^2 \lambda} \\
& = \frac{1}{2}\left[ -\frac{1}{\xi + \theta \zeta} +  2\sum_{n=1}^\infty \frac{\xi + \theta \zeta}{(\pi n)^2-(\xi + \theta \zeta)^2} \right] + \frac{1}{2}\left[ -\frac{1}{\xi - \theta \zeta} +  2\sum_{n=1}^\infty \frac{\xi - \theta \zeta}{( \pi n)^2-(\xi - \theta \zeta)^2} \right] 
\end{split}
\end{equation*}
and 
\begin{equation*}
\begin{split}
G(\zeta,\xi,\lambda) & = \frac{\zeta}{\xi^2- \lambda \zeta^2} + 2\sum_{n=1}^{\infty}G_n(\zeta,\xi,\lambda) \\
& = \frac{\zeta}{\xi^2- \lambda \zeta^2} + \sum_{n=1}^{\infty} \frac{2 \zeta[(\pi n)^2+(\xi^2-\lambda \zeta^2)]}{[(\pi n)^2-(\xi^2+\lambda \zeta^2)]^2-(2 \zeta \xi)^2 \lambda} \\
& = \frac{1}{2\theta} \left[-\frac{1}{\xi + \theta \zeta} + 2 \sum_{n=1}^\infty \frac{\xi + \theta \zeta}{(\pi n)^2-(\xi + \theta \zeta)^2}  \right] -\frac{1}{2\theta} \left[-\frac{1}{\xi - \theta \zeta} + 2 \sum_{n=1}^\infty \frac{\xi - \theta \zeta}{(\pi n)^2-(\xi - \theta \zeta)^2}  \right],
\end{split}
\end{equation*}
where implicitly we have extended $\mathbb{R}$ to the ring $\mathbb{R}_\lambda$ (\ref{Rlambda}). \\\\
As shown in Appendix D, using again the Weierstrass factorization theorem and properties of $\mathbb{R}_\lambda$, the functions $F(\zeta,\xi,\lambda)$ and $G(\zeta,\xi,\lambda)$ over $\mathbb{R}_\lambda$ are  given by the following functions for the different signs of $\lambda$: 
\begin{itemize}
    \item For $\lambda=0$, then      \begin{subequations}
  \begin{align}
F(\zeta,\xi,0) & = \tan \left(\xi-\frac{\pi}{2} \right), \label{eqf:1}  \\
G(\zeta,\xi,0) & = \frac{\zeta}{2 \cos^2 \left(\xi- \frac{\pi}{2} \right)}  \label{eqg:2}
\end{align}
\end{subequations}
\item For $\lambda \neq 0$, then  
\begin{subequations}
\begin{align}
F(\zeta,\xi,\lambda) &  = -\frac{1}{2} \left[ \cot \left( \xi+ \theta \zeta  \right) + \cot \left( \xi - \theta \zeta \right) \right], \label{eqFF:1}  \\
G(\zeta,\xi,\lambda) & = -\frac{1}{2 \theta} \left[ \cot \left( \xi+ \theta \zeta \right) - \cot \left( \xi- \theta \zeta \right) \right]  \label{eqFF:2}
\end{align}
\end{subequations}
which for $\lambda<0$ are given by
\begin{subequations}
\begin{align}
F(\zeta,\xi,\lambda) & =  \frac{-\sin(2\xi)}{-\cos(2\xi)+ \cosh(2\sqrt{|\lambda|}\zeta)}, \\
G(\zeta,\xi,\lambda) & = \frac{1}{ \sqrt{|\lambda|}} \frac{\sinh(2 \sqrt{|\lambda|}\zeta)}{-\cos(2\xi)+ \cosh(2\sqrt{|\lambda|}\zeta)},
\end{align}
\end{subequations}
and for $\lambda>0$
\begin{subequations}
\begin{align}
F(\zeta,\xi,\lambda) & =  \frac{-\sin(2 \xi)}{-\cos (2\xi)+\cos(2\sqrt{\lambda}\zeta)}, \\
G(\zeta,\xi,\lambda) & = \frac{1}{ \sqrt{\lambda}} \frac{\sin(2 \sqrt{\lambda} \zeta)}{-\cos(2\xi)+ \cos(2\sqrt{\lambda}\zeta)}.
\end{align}
\end{subequations}
\end{itemize}
Finally, replacing (\ref{mat}) in (\ref{FGP}), we obtain the skew-symmetric part of the canonical FGMPP dynamical $r$-matrix associated to $(\mathfrak{g}_\lambda,\mathfrak{h}_{\lambda,0},K_{\alpha \beta})$ is given by 
\begin{equation*}
\begin{split}
\rho_{\text{FGMPP}}(x_0)& = \langle T_a, g(\text{ad}[K_{\alpha \beta}^{\lor}(x_0)])J_1    \rangle_s T^a \otimes P^1 + \langle T_a, g(\text{ad}[K_{\alpha \beta}^{\lor}(x_0)])J_2    \rangle_s T^a \otimes P^2 \\
& \quad + \langle T_a, g(\text{ad}[K_{\alpha \beta}^{\lor}(x_0)])P_1    \rangle_s T^a \otimes J^1 + \langle T_a, g(\text{ad}[K_{\alpha \beta}^{\lor}(x_0)])P_2    \rangle_s T^a \otimes J^2 \\
& = \tilde{b}(\zeta,\xi,\lambda)(J^1 \otimes P^2 -J^2 \otimes P^1 - P^2 \otimes J^1 + P^1 \otimes J^2) \\
& \qquad + \tilde{c}(\zeta,\xi,\lambda)(P^1 \otimes P^2 - P^2 \otimes P^1 -\lambda J^2 \otimes J^1 + \lambda J^1 \otimes J^2)  
\end{split}
\end{equation*}
where 
\begin{equation}\label{btildecoeflong}
    \tilde{b}(\zeta,\xi,\lambda)  = \frac{\alpha}{\alpha^2-\lambda \beta^2} F(\zeta,\xi,\lambda) -\lambda \frac{\beta}{\alpha^2-\lambda \beta^2} G(\zeta,\xi,\lambda)
\end{equation}
and 
\begin{equation}\label{ctildecoeflong}
    \tilde{c}(\zeta,\xi,\lambda)  = \frac{\alpha}{\alpha^2-\lambda \beta^2} F (\zeta,\xi,\lambda)- \frac{\beta}{\alpha^2-\lambda \beta^2} G (\zeta,\xi,\lambda)
\end{equation}
concluding the \textit{canonical} FGMPP dynamical $r$-matrix associated to $(\mathfrak{g}_\lambda,\mathfrak{h}_{\lambda,0},K_{\alpha \beta})$ is given by
\begin{equation}\label{FGMPPans}
r_{\text{FGMPP}}(x_0)= K_{\alpha \beta} + \tilde{b}(\xi,\zeta,\lambda)(P^1 \wedge J^2 - P^2 \wedge J^1) + \tilde{c}(\xi,\zeta,\lambda)(\lambda J^1 \wedge J^2 + P^1 \wedge P^2),
\end{equation}
which coincides precisely with the solution $r_d(\psi_0,\gamma_0)$ in Theorem \ref{MainTheo2} by choosing $\Psi_0=\frac{\pi}{2} \alpha$ and $\Gamma_0= \frac{\pi}{2} \beta$. \\\\
In a completely analogous way if we consider a generic element $x_1 =\gamma_1 J_1^* + \psi_1 P_1^*$ in $\mathfrak{h}^*_{\lambda,1}$, we obtain for $g(\text{ad}[K_{\alpha \beta}^{\lor}(x_1)]|_{\mathfrak{h}^\perp_{\lambda,1}})$ a matrix similar to (\ref{mat}) but with $F$ and $G$ now given by
\begin{subequations}
  \begin{align}
F(\zeta,\xi,0) & = \tanh \left(\xi - \frac{\pi}{2}  \right), \label{eqfL:1}  \\
G(\zeta,\xi,0) & = \frac{\zeta}{2 \cosh^2 \left(\xi- \frac{\pi}{2} \right)}  \label{eqgL:2}
\end{align}
\end{subequations}
for $\lambda=0$,       
\begin{subequations}
\begin{align}
F(\zeta,\xi,\lambda) & =  \frac{\sinh(2\xi-\pi)}{\cosh(2\xi-\pi)+ \cosh(2\sqrt{|\lambda|}\zeta)}, \\
G(\zeta,\xi,\lambda) & = \frac{1}{ \sqrt{|\lambda|}} \frac{\sin(2 \sqrt{|\lambda|}\zeta)}{\cosh(2\xi-\pi)+ \cos(2\sqrt{|\lambda|}\zeta)},
\end{align}
\end{subequations}
for $\lambda<0$, and
\begin{subequations}
\begin{align}
F(\zeta,\xi,\lambda) & =  \frac{\sinh(2 \xi-\pi)}{\cosh (2\xi-\pi)+\cosh(2\sqrt{\lambda}\zeta)}, \\
G(\zeta,\xi,\lambda) & = \frac{1}{ \sqrt{\lambda}} \frac{\sinh(2 \sqrt{\lambda} \zeta)}{\cosh(2\xi-\pi)+ \cosh(2\sqrt{\lambda}\zeta)}.
\end{align}
\end{subequations}
for $\lambda>0$. \\\\Thus, evaluating (\ref{FGP}) in this case leads to a classical dynamical $r$-matrix of the form 
\begin{equation}\label{FGMPPans1}
r_{\text{FGMPP}}(x_1)= K_{\alpha \beta} + \tilde{b}(\xi,\zeta,\lambda)(P^2 \wedge J^0 - J^2 \wedge P^0) + \tilde{c}(\xi,\zeta,\lambda)(\lambda J^2 \wedge J^0 + P^2 \wedge P^0),
\end{equation}
with $\tilde{b}$ and $\tilde{c}$ defined exactly as before, which coincides with the solution $r_d(\psi_1,\gamma_1)$ in Theorem \ref{MainTheo2} by choosing again $\Psi_0=\frac{\pi}{2} \alpha$ and $\Gamma_0=\frac{\pi}{2} \beta$.  
\end{proof}
\section{ Conclusion and outlook}

In this paper we gave a full classification of the classical dynamical $r$-matrices up to gauge equivalence  for the Lie algebras $\mathfrak{g}_\lambda$ of the local isometry groups of the maximally symmetric Euclidean and Lorentzian spaces in three dimensions.  The classical dynamical $r$-matrices for vanishing cosmological constant in the Lorentzian setting were obtained previously by Meusburger and Sch\"{o}nfeld in \cite{Meusburger:2012wc}, but our results for  classical dynamical $r$-matrices for the general triple ($\mathfrak{g}_\lambda$, $\mathfrak{h}_\lambda$, $K_{\alpha \beta}$) are new. 

It is interesting to compare our results with the analogous study of (non-dynamical) classical $r$-matrices in \cite{Osei:2017ybk} for the family  $\mathfrak{g}_\lambda$. The general ansatz used here is a generalisation  the one used in that paper, where three families of solutions were identified, associated with the Lie bialgebra structure of a classical double, a generalised bicrossproducct (or $\kappa$-Poincar\'e algebra) or a generalised complexification (in the  same sense as used in this paper) of the standard $\mathfrak{sl}(2,\mathbb{R})$ bialgebra structure. Here we found that, of these three, only the dynamical version of the generalised complexification solves both the dynamical classical Yang-Baxter equation and the equivariance condition required for dynamical $r$-matrices.  In particular,  there is no dynamical version of the $r$-matrices associated to the family of classical double structures identified in \cite{Osei:2017ybk}. It is interesting  that the  equivariance constraint  effectively makes the  full classification of classical dynamical $r$-matrices for $\mathfrak{g}_\lambda$ easier than the classification of non-dynamical $r$-matrices for that family of Lie algebras (which indeed has not yet been achieved). 

While our results can be summarised conveniently as generalised complexifications of known results for  $\mathfrak{sl}(2,\mathbb{R})$ or $\mathfrak{su}(2)$, our derivations and proofs do not make essential use of this point of view. In future work, it may be interesting to see  which parts of standard holomorphic function theory extend when $\mathbb{C}$ is replaced by the ring $\mathbb{R}_\lambda$, and if the resulting machinery is sufficient to establish our results directly by `generalised holomorphic methods'.  It would then also be interesting if this point of view can be usefully adopted in the quantization.

As remarked, the quantization of the Hamiltonian Chern-Simons theory (i.e. over manifolds homeomorphic to $\mathbb{R} \times \Sigma_{g,n}$) reduces to the quantization of a constrained Poisson space $\mathcal{P}_{\text{ext}}^{g,n}$, whose Poisson structure is defined in terms of a classical $r$-matrix $r$. In the so-called combinatorial quantization approach (see \cite{alekseev1995combinatorial} and \cite{alekseev1996combinatorial}) the quantization is performed before imposing the constraints. It  relies on  the existence  of a quantum $R$-matrix  which quantizes $r$, i.e.  which has an expansion
\begin{equation*}
    R(\hbar) = 1 + \hbar r + \mathcal{O}(\hbar^2).
\end{equation*}
Meanwhile in the approach presented here, the constraints are imposed before quantization, leading to the appearance of a classical dynamical $r$-matrix $r^d$ describing the Poisson structure of the (constrained) space $\mathcal{P}^{g,n}$. Hence a dynamical combinatorial quantization (i.e. where the classical $r$-matrix is replaced by a classical dynamical $r$-matrix $r^d$, see \cite{Spies2020}) is expected to lead to a quantization of the Hamiltonian Chern-Simons theory. In analogy to the non-dynamical case, this quantization scheme requires  a quantum dynamical $R$-matrix that quantizes $r^d$. In our case, where the Poisson structure is defined in terms of classical dynamical $r$-matrices gauge equivalent to the  classical dynamical $r$-matrix, the quantization requires the existence of a quantum dynamical $R$-matrix, i.e.       
\begin{equation*}
R(x,\hbar)= 1 + \hbar r(x) + \mathcal{O}(\hbar^2).
\end{equation*}
The quantization of the AM classical dynamical $r$-matrix discussed at the start of Section \ref{Sec4paper1},  was addressed by Enriquez and Etingof in \cite{Enriquez2003QuantizationOA} and the extension of this  approach to the quantization of the FGMPP classical dynamical $r$-matrix seems possible and interesting.  Further developments in this direction could lead to a complete formulation of quantum character varieties, with potentially interesting links to  the categorical approaches \cite{ben2018quantum} and \cite{keller:tel-04062148}), and a formal understanding of quantum of  3d gravity (see \cite{Buffenoir:2005zi}). Also, extensions of the results presented here to the setting of  supersymmetric Chern-Simons  theory (see \cite{Mikovic:2000cx}) and Super Character Varieties (see \cite{Aghaei:2018cbn}) seem interesting. For the family of Lie algebras $\mathfrak{g}_\lambda$ relevant for 3d gravity, the point of view of generalised complexification, which proved so useful here, may provide additional  insights and connections.


%

\textbf{Acknowledgments.} We are grateful to Catherine Meusburger for helpful discussions and explanations regarding the physical interpretation of the classical dynamical $r$-matrices in the setting of 3d  Chern-Simons gravity. We would also like to thank her for sharing with us unpublished results concerning extensions (for $\Lambda_C \neq 0$ and both signatures) of classical dynamical $r$-matrices derived previously in \cite{Meusburger:2012wc}, which coincide with (\ref{Cathshar0})-(\ref{Cathshar0}) for $\alpha=2$ and $\psi_C=\gamma_C=0$. The work of J.C.M.P is supported by a James Watt scholarship from Heriot-Watt University. 

\section*{Appendix A: Action of \texorpdfstring{$G$}{} over the space \texorpdfstring{$\text{Dyn}(\mathfrak{g},K)$}{}}\label{AppA} 
Let $\mathfrak{g}$ be a Lie algebra and $K \in (S^2 \mathfrak{g})^{\mathfrak{g}}$. In this Appendix we prove that if $\mathfrak{h}$ and $\mathfrak{h}'$ are two conjugate Lie subalgebras of $\mathfrak{g}$, say $\mathfrak{h}'=\text{Ad}(g)(\mathfrak{h})$ for $g \in G$, then the sets of classical dynamical $r$-matrices $\text{Dyn}(\mathfrak{g},\mathfrak{h},K)$ and $\text{Dyn}(\mathfrak{g},\mathfrak{h}',K)$ are in bijection, via a natural map constructed using the adjoint action of $G$ over $\mathfrak{g} \otimes \mathfrak{g}$ and its coadjoint action over $\mathfrak{g}^*$.  \\\\ 
\textbf{Lemma A.} Let $\mathfrak{h}$ and $\mathfrak{h}'$ be conjugate Lie subalgebras of $\mathfrak{g}$, say by $g \in G$. Then the map 
\begin{equation}\label{mapdynfin}
\begin{split}
    \text{Dyn}(\mathfrak{g},\mathfrak{h},K)  & \to \text{Dyn}(\mathfrak{g},\mathfrak{h}',K)  \\
    r(x) & \mapsto \text{Ad}(g) \otimes \text{Ad}(g) r(\text{Ad}^*(g^{-1})x)
    \end{split}
\end{equation}
is well-defined and bijective. \\\\
\begin{proof}
Take $r \in \text{Dyn}(\mathfrak{g},\mathfrak{h},K)$ and let $g \in G$ such that $\text{Ad}(g)(\mathfrak{h})=\mathfrak{h}'$. Define 
\begin{equation*}
    r'(x)=\text{Ad}(g) \otimes \text{Ad}(g)r(\text{Ad}^*(g^{-1})x) \qquad \text{for } x \in (\mathfrak{h}')^* 
\end{equation*}
where $\text{Ad}^*: G \to \text{Aut}(\mathfrak{g}^*)$ is the coadjoint action of the Lie group $G$ over the dual of the Lie algebra $\mathfrak{g}^*$. \\\\
Writing $r$ in terms of a basis $\{T_a\}_{a=1,\cdots, \dim \mathfrak{g}}$ of $\mathfrak{g}$, we get $r(x)=r^{ab}(x)T_a \otimes T_b$ for all $x \in \mathfrak{h}^*$ and by direct computation we obtain
\begin{equation*}
\begin{split}
    [r'_{12}(x),r'_{23}(x)] & = [r^{ab}(\text{Ad}^*(g^{-1})x) \text{Ad}(g)T_a \otimes \text{Ad}(g)T_b \otimes 1, r^{cd}(\text{Ad}^*(g^{-1})x) 1 \otimes \text{Ad}(g)T_c \otimes \text{Ad}(g)T_d] \\
    & = r^{ab}(\text{Ad}^*(g^{-1})x)r^{cd}(\text{Ad}^*(g^{-1})x) \text{Ad}(g)T_a \otimes [\text{Ad}(g)T_b,\text{Ad}(g)T_c] \otimes \text{Ad}(g)T_d \\
    & = \text{Ad}(g)^{\otimes 3} [r_{12}(\text{Ad}^*(g^{-1})x),r_{23}(\text{Ad}^*(g^{-1})x)].
    \end{split}
\end{equation*}
Similarly 
\begin{equation*}
\begin{split}
[r'_{13}(x),r'_{23}(x)] & =\text{Ad}(g)^{\otimes 3} [r_{13}(\text{Ad}^*(g^{-1})x),r_{23}(\text{Ad}^*(g^{-1})x)], \\  
[r'_{12}(x),r'_{13}(x)] & =\text{Ad}(g)^{\otimes 3} [r_{12}(\text{Ad}^*(g^{-1})x),r_{13}(\text{Ad}^*(g^{-1})x)]    
\end{split}
\end{equation*}
and consequently
\begin{equation}\label{Brackinv}
\begin{split}
    [[r'(x),r'(x)]] & = \text{}[[\text{Ad}(g) \otimes \text{Ad}(g)r(\text{Ad}^*(g^{-1})x),\text{Ad}(g) \otimes \text{Ad}(g)r(\text{Ad}^*(g^{-1})x)]]  \\
    & = \text{Ad}(g)^{\otimes 3}  [[r(\text{Ad}^*(g^{-1})x),r(\text{Ad}^*(g^{-1})x)]].
\end{split}
\end{equation}
Also, by direct computation 
\begin{equation*}
\begin{split}
    {h'_i}^{(1)} \frac{\partial r'_{23}}{\partial h'^i} & = (\text{Ad}(g)h_i)^{(1)} (\text{Ad}(g) \otimes \text{Ad}(g)) \frac{\partial r_{23}}{\partial h^i}(\text{Ad}^*(g^{-1})x) \\
    & = \text{Ad}(g)^{\otimes 3} h_i^{(1)}\frac{\partial r_{23}}{\partial h^i}(\text{Ad}^*(g^{-1})x)
\end{split}
\end{equation*}
and analogously
\begin{equation*}
    {h'_i}^{(2)} \frac{\partial r'_{13}}{\partial h'^i} = \text{Ad}(g)^{\otimes 3} h_i^{(2)}\frac{\partial r_{13}}{\partial h^i}(\text{Ad}^*(g^{-1})x), \qquad
    {h'_i}^{(3)} \frac{\partial r'_{12}}{\partial h'^i} = \text{Ad}(g)^{\otimes 3} h_i^{(3)}\frac{\partial r_{12}}{\partial h^i}(\text{Ad}^*(g^{-1})x).
\end{equation*}
Hence
\begin{equation*}
\begin{split}
\text{Alt}(dr')(x) & =  \sum_{i = 1}^{\dim \mathfrak{h}'}  {h'_i}^{(1)} \frac{\partial r'_{23}}{\partial h'^i}(x) - {h'_i}^{(2)} \frac{\partial r'_{13}}{\partial h'^i}(x) + {h'_i}^{(3)} \frac{\partial r'_{12}}{\partial h'^i}(x) \\
& = \text{Ad}(g)^{\otimes 3} \left[ h_i^{(1)}\frac{\partial r_{23}}{\partial h^i}(\text{Ad}^*(g^{-1})x) - h_i^{(1)}\frac{\partial r_{23}}{\partial h^i}(\text{Ad}^*(g^{-1})x) + h_i^{(1)}\frac{\partial r_{23}}{\partial h^i}(\text{Ad}^*(g^{-1})x) \right] \\
& = \text{Ad}(g)^{\otimes 3} \text{Alt}(dr)(\text{Ad}^*(g^{-1})x)
\end{split}
\end{equation*}
and so, by combining with (\ref{Brackinv}), we conclude $r'$ is indeed a solution of the CDYBE, i.e.
\begin{equation*}
\begin{split}
\text{CDYB}(r')(x) & =[[r'(x),r'(x)]]+ \text{Alt}(dr)(x) \\
& = \text{Alt}(g)^{\otimes 3}([[r(\text{Ad}^*(g^{-1})x),r(\text{Ad}^*(g^{-1})x)]]+ \text{Alt}(dr)(\text{Ad}^*(g^{-1})x)) \\
& = \text{Alt}(g)^{\otimes 3} (\text{CDYBE}(r)(\text{Ad}^*(g^{-1})x))   \\
& = -\text{Alt}(g)^{\otimes 3}[[K,K]] \\
& = -[[K,K]].
\end{split}
\end{equation*}
Since the subspace $S^2(\mathfrak{g})$ of $\mathfrak{g} \otimes \mathfrak{g}$ is invariant under the Adjoint action of the Lie group $G$, the fact $\text{Sym}(r)=K \in (S^2 \mathfrak{g})^{\mathfrak{g}}$ implies  
\begin{equation*}
\text{Sym}(r')=\text{Sym}(r)=K.  
\end{equation*}
Finally, take $h' \in \mathfrak{h}'$ and let $h \in \mathfrak{h}$ be such that $h'=\text{Ad}(g)h$. Then we get 
\begin{equation*}
\begin{split}
\frac{d }{ds}\bigg|_{s=0} r'( \text{Ad}^*(e^{sh'}) x)  & = \text{Ad}(g) \otimes \text{Ad}(g) \frac{d }{ds}\bigg|_{s=0} r( \text{Ad}^*(g^{-1}) \text{Ad}^*(e^{s h'}) x) \\
& = \text{Ad}(g) \otimes \text{Ad}(g) \frac{d }{ds}\bigg|_{s=0} r( \text{Ad}^*(e^{sh})\text{Ad}^*(g^{-1})x), \\
\end{split}
\end{equation*}
where we have used $\text{Ad}^*(e^{sh'})=\text{Ad}^*(g)\text{Ad}^*(e^{sh})\text{Ad}^*(g^{-1})$, and 
\begin{equation*}
\begin{split}
[r'(x),h' \otimes 1 + 1 \otimes h'] & =  [\text{Ad}(g)\otimes \text{Ad}(g) r(\text{Ad}^*(g^{-1})x), \text{Ad}(g)h \otimes 1 + 1 \otimes \text{Ad}(g)h] \\
& = \text{Ad}(g) \otimes \text{Ad}(g) [r(\text{Ad}^*(g^{-1})x), h \otimes 1 + 1 \otimes h].
\end{split}
\end{equation*}
Therefore
\begin{equation*}
    \frac{d }{ds}\bigg|_{s=0} r'( \text{Ad}^*(e^{sh'}) x) +  [r'(x),h' \otimes 1 + 1 \otimes h'] 
 \end{equation*}
equals
 \begin{equation*}
 \text{Ad}(g) \otimes \text{Ad}(g) \left( \frac{d }{ds}\bigg|_{s=0} r( \text{Ad}^*(e^{sh})\text{Ad}^*(g^{-1})x) + [r(\text{Ad}^*(g^{-1})x), h \otimes 1 + 1 \otimes h] \right) =0,      
 \end{equation*}
i.e. the $\mathfrak{h}'$-equivariance of $r'$. \\\\
The previous computations show the map (\ref{mapdynfin}) is well-defined and is clearly bijective (the inverse is obtained  by replacing $g$ with $g^{-1}$). 
\end{proof}
\section*{Appendix B: Conjugacy classes of Cartan Subalgebras of \texorpdfstring{$\mathfrak{g}_\lambda$}{}}\label{AppB}
One of the main results of this paper is the complete description of the set $\text{Dyn}^\mathfrak{C}(\mathfrak{g}_\lambda,K_{\alpha \beta})$ (Theorem \ref{MainTheo2}) and the moduli space $\mathcal{M}^{\mathfrak{C}}(\mathfrak{g}_\lambda,K_{\alpha \beta})$ (Lemma \ref{Lemma3}) of Cartan classical dynamical $r$-matrices associated to $(\mathfrak{g}_\lambda,K_{\alpha \beta})$. Lemma A is essential to achieve this since reduces the task to finding classical dynamical $r$-matrices for representatives of $\mathfrak{C}_{\mathfrak{g}_\lambda}^{\text{Ad}}$ (the conjugacy classes of Cartan subalgebras of $\mathfrak{g}_\lambda$). \\\\ As stated in Subsection 3.1 the cardinality of $\mathfrak{C}_{\mathfrak{g}_\lambda}^{\text{Ad}}$ is at most two for all $\lambda \in \mathbb{R}$. More precisely, in the case where the cardinal is one all the Cartan subalgebras are conjugate equivalent to $\mathfrak{h}_{\lambda,0}$ (the Cartan subalgebra generated by $J_0$ and $P_0$), while in the case it is two they could be conjugate equivalent to $\mathfrak{h}_{\lambda,0}$ or $\mathfrak{h}_{\lambda,1}$ (the Cartan subalgebra generated by $J_1$ and $P_1$). \\\\
In this Appendix we provide a detailed description of the Cartan subalgebras of the simple ($\mathfrak{so}(3,1)$), semisimple ($\mathfrak{so}(4)$ and $\mathfrak{so}(2,2)$) and indecomposable non-solvable ($\mathfrak{iso}(3)$ and $\mathfrak{iso}(2,1)$) six dimensional real Lie algebras.  
\subsection*{Simple case: \texorpdfstring{$\mathfrak{so}(3,1)$}{}}
By $\mathfrak{so}(3,1)_\lambda$ ,with $\lambda<0$, we denote the six dimensional real algebra generated by $\{J_0,J_1,J_2,P_0,P_1,P_2\}$ such that
\begin{equation*}
    [J_a,J_b]=\epsilon_{abc}J^c, \qquad [J_a,P_b]= \epsilon_{abc}P^c, \qquad [P_a,P_b]=\lambda \epsilon_{abc}J^c.
\end{equation*}
All these Lie algebras are isomorphic to $\mathfrak{so}(3,1)_{-1}$ via the map $J_a \to J_a$ and $P_a \to \sqrt{|\lambda|}P_a$. In standard literature $\mathfrak{so}_{-1}(3,1)$ is simply denoted by $\mathfrak{so}(3,1)$ and is the only (up to isomorphism) simple six dimensional real Lie algebra. This follows from the fact $\mathfrak{so}(3,1)$ is isomorphic to $\mathfrak{sl}(2,\mathbb{C})$ (viewed as a real Lie algebra of dimension six). \\\\
The Lie algebra $\mathfrak{so}(3,1)$ has a representation as linear maps  $\text{End}(\mathbb{R}^4)$, via the traceless matrices of the form
\begin{equation}
    m^iJ_i + s^iP_i = \begin{pmatrix} 0 & -m^0 & -m^1 & s^2  \\ m^0 & 0 & m^2 & s^1 \\ m^1 & -m^2 & 0 & s^0 \\ s^2 & s^1 & s^0  & 0 \end{pmatrix}.
\end{equation}
It is straightforward to check that if the eigenvalues of $m^iJ_i+s^i P_i$ are \textit{all} zero then there exists $M \in \text{SO}(3,1)$ such that 
\begin{equation}
M(m^i J_i + s^i P_i)M^{-1} \in  \mathbb{R}_{\neq 0}(P_2-J_0),
\end{equation}
while if \textit{at least one} is not zero there exists $\tilde{M} \in \text{SO}(3,1)$ such that
\begin{equation}\label{selfnorm}
\tilde{M}(m^i J_i + s^i P_i)\tilde{M}^{-1} \in \{m J_0 + s P_0 \ | \ (m,s) \in \mathbb{R}^2, (m,s) \neq (0,0) \}.
\end{equation}
Clearly the Lie subalgebra spanned by $J_2$ and $P_2$ is nilpotent (indeed it is abelian) and by (\ref{selfnorm}) self-normalizing, then we conclude this Lie subalgebra is a Cartan subalgebra of $\mathfrak{so}(3,1)$ and so $\text{rnk}(\mathfrak{so}(3,1))$ is two. \\\\ 
Since the two dimensional Lie subalgebras of $\mathfrak{so}(3,1)$ are Abelian and the normalizer of $P_2-J_0$ is spanned by $P_0-J_2$ and $P_1$, we get that there exist only one conjugacy class of Cartan subalgebras of $\mathfrak{so}(3,1)$ with representative $\text{Span}\{J_0,P_0\}$.

\subsection*{Semisimple cases: \texorpdfstring{$\mathfrak{so}(4)$}{} and \texorpdfstring{$\mathfrak{so}(2,2)$}{}}
In the case when $\lambda>0$, the Lie algebras $\mathfrak{g}_\lambda$
  \begin{equation*}
    [J_a,J_b]=\epsilon_{abc}J^c, \qquad [J_a,P_b]= \epsilon_{abc}P^c, \qquad [P_a,P_b]=\lambda \epsilon_{abc}J^c
\end{equation*}
result to be semi-simple for both the Euclidean and Lorentzian cases. To see this we consider the following six alternative generators 
\begin{equation*}
J_a^{\pm}= \frac{1}{2} \left(J_a \pm \frac{1}{\sqrt{\lambda}}P_a \right) \qquad \text{for } a=0,1,2.
\end{equation*}
Since this new set of generators satisfies the commutation relations 
\begin{equation*}
    [J^{\pm}_a,J^{\pm}_b]= \epsilon_{abc}(J^{\pm})^c, \qquad [J^{+}_a,J^{-}_b]=0, \qquad \text{for } a,b=0,1,2, 
\end{equation*}
then we conclude in the Euclidean case
\begin{equation}
    \mathfrak{so}(4) \cong \mathfrak{su}(2) \oplus \mathfrak{su}(2),
\end{equation} 
while in the Lorentzian case 
\begin{equation}
    \mathfrak{so}(2,2) \cong \mathfrak{sl}(2,\mathbb{R}) \oplus \mathfrak{sl}(2,\mathbb{R});
\end{equation} 
showing explicitly that for $\lambda>0$ the Lie algebras $\mathfrak{g}_\lambda$ could be factorised as the direct sum of two simple three dimensional Lie subalgebras. \\\\
\textbf{Proposition A} \cite{Sugiura1959ConjugateCO}\textbf{.} Let $\mathfrak{g}$ be a semisimple real Lie algebra, with simple decomposition given by
\begin{equation*}
    \mathfrak{g}=  \oplus_{i=1}^n \mathfrak{g}_n 
\end{equation*}
and $\mathfrak{h}$ a Cartan subalgebra of $\mathfrak{g}$. Then $\mathfrak{h}$ is decomposable as a direct sum 
\begin{equation*}
    \mathfrak{h}= \oplus_{i=1}^n \mathfrak{h}_n 
\end{equation*}
such that
\begin{enumerate}
    \item $\mathfrak{h}_i= \mathfrak{h} \cap \mathfrak{g}_i$ for all $i=1,\cdots,n$.
    \item $\mathfrak{h}_i$ is a Cartan subalgebra of the simple Lie algebra $\mathfrak{g}_i$ for $i=1,\cdots,n$.
\end{enumerate}

This proposition indicates that the set of Cartan subalgebras of a semisimple algebra $\mathfrak{g}$ is contained in the set of subalgebras generated by taking the direct sum of Cartan subalgebras of the simple factors of $\mathfrak{g}$. Therefore, since the adjoint actions of $\text{SO}(4)$ and $\text{SO}(2,2)$ factor through the actions of each copy of $\text{SU}(2)$ and $\text{SL}(2,\mathbb{R})$ over each factor of $\mathfrak{su}(2) \oplus \mathfrak{su}(2)$ and $\mathfrak{sl}(2,\mathbb{R}) \oplus \mathfrak{sl}(2,\mathbb{R})$, respectively, the conjugacy classes of Cartan subalgebras of $\mathfrak{so}(4)$ and $\mathfrak{so}(2,2)$ are given by the direct sum of conjugacy classes of Cartan subalgebras of $\mathfrak{su}(2)$ and $\mathfrak{sl}(2,\mathbb{R})$, respectively.      
Hence, the problem reduces to determining  first the conjugacy classes of Cartan subalgebras of real forms of $\mathfrak{sl}(2,\mathbb{C})$, which are known to have rank one: in the case of $\mathfrak{su}(2)$ there exists just one conjugacy class of Cartan subalgebras with representative spanned by $J_0$, while $\mathfrak{sl}(2,\mathbb{R})$ has two conjugacy classes of Cartan subalgebras with representatives spanned by $J_0$ and $J_1$.  \\\\
In the Euclidean case there is clearly  just one conjugacy class of Cartan Lie subalgebras with  
\begin{equation*}
    \text{Span}\{J_0^+\} \oplus \text{Span}\{J_0^-\}  \cong \text{Span}\{ J_0, P_0 \},
\end{equation*}
as a representative. Meanwhile, in the Lorentzian  one can show that  there exist four conjugacy classes of Cartan subalgebras, with 
\begin{equation*}
    \text{Span} \{J_0^+\}  \oplus \text{Span} \{J_0^- \}  \cong \text{Span}\{J_0, P_0 \},
    \end{equation*}
\begin{equation*}
   \text{Span} \{J_1^+\}  \oplus \text{Span} \{J_1^- \} \cong \text{Span}\{J_1, P_1 \},
\end{equation*}
\begin{equation*}
   \text{Span} \{J_0^+\}  \oplus \text{Span} \{J_1^- \} \cong \text{Span}\{J_0+P_0, J_1-P_1 \},
\end{equation*}
and
\begin{equation*}
   \text{Span} \{J_1^+\}  \oplus \text{Span} \{J_0^- \} \cong \text{Span}\{J_1+P_1, J_0-P_0 \}
\end{equation*}
as representatives.


\subsection*{Semidirect sum cases: \texorpdfstring{$\mathfrak{iso}(3)$}{} and \texorpdfstring{$\mathfrak{iso}(2,1)$}{}}
Finally in the $\lambda=0$, the Lie algebras $\mathfrak{g}_\lambda$
 \begin{equation*}
    [J_a,J_b]=\epsilon_{abc}J^c, \qquad [J_a,P_b]= \epsilon_{abc}P^c, \qquad [P_a,P_b]=0
\end{equation*}
are isomorphic to the semidirect sums $\mathfrak{so}(3) \ltimes_{\text{ad}^*} \mathfrak{so}(3)^*$ and $\mathfrak{sl}(2,\mathbb{R}) \ltimes_{\text{ad}^*} \mathfrak{sl}(2,\mathbb{R})^*$ for the Euclidean and Lorentzian cases, respectively.\\\\
\textbf{Proposition B} \cite{1d57cda6-5aff-3896-b816-ef14376cd807}\textbf{.} Let $\mathfrak{k}$ be a Lie algebra that acts over a nilponent Lie algebra $V$ via $\phi: \mathfrak{k} \to \text{Der}(V)$. If $\mathfrak{h}$ is a Cartan subalgebra of $\mathfrak{k}$, then 
\begin{equation*}
    \mathfrak{h} \ltimes_\phi V^0(\mathfrak{h})
\end{equation*}
is a Cartan subalgebra of the Lie algebra $\mathfrak{g} \equiv \mathfrak{k} \ltimes_\phi V$, where 
\begin{equation*}
    V^0(\mathfrak{h})\equiv \{v \in V \ | \ \forall x \in \mathfrak{h}, \ \exists n \in \mathbb{Z}_{>0} \text{ such that } \phi(x)^nv=0 \}.  
\end{equation*}

For our cases of interest, $\mathfrak{k}=\mathfrak{so}(3)$ or $\mathfrak{k}=\mathfrak{sl}(2,\mathbb{R})$ with $V=\mathbb{R}^3 \cong \mathfrak{so}(3)^*$ or $V=\mathbb{R}^{1,2} \cong \mathfrak{sl}(2,\mathbb{R})^*$, respectively, such that the action $\phi$ is given by the coadjoint action, i.e. 
\begin{equation*}
\begin{split}
\text{ad}^*: \mathfrak{k} & \to \text{Der}(\mathfrak{k}^*) \\
x & \mapsto [x,\cdot]
\end{split}
\end{equation*}
in both cases. \\\\
From this proposition it follows immediately
\begin{equation*}
   \langle  J_i \rangle \ltimes_{\text{ad}^*} \langle P_i \rangle  \qquad \text{ for } i=0,1,2
\end{equation*}
are Cartan subalgebras of $\mathfrak{iso}(3)$ and $\mathfrak{iso}(2,1)$, showing in this way both Lie algebras have rank two. \\\\
\textbf{Proposition C.} The Cartan subalgebras of the Lie algebras $\mathfrak{iso}(3)$ and $\mathfrak{iso}(2,1)$ are generated by sets of the form 
    \begin{equation}\label{cartaniso2}
    \{m_aP^a, m_a J^a + n_a P^a  \ | \ \bm{m}, \bm{n} \in \mathbb{R}^{1,2}, \bm{m}^2 \neq 0, \bm{m} \cdot \bm{n} = 0  \}  
    \end{equation}
\begin{proof}
Start considering the most general set of two elements in $\mathfrak{g}_\lambda$ i.e. 
\begin{equation}\label{basis}
    \{ \bm{m}\bm{P} + \bm{n}\bm{J},  \bm{k}\bm{P}+\bm{\ell}\bm{J} \}
\end{equation}
where we use the notation $\bm{m}\bm{P}$ ($\bm{\ell}\bm{J}$) to denote in a compact manner the elements $m_aP^a$ ($\ell _aJ^a$) in the Lie algebra. \\\\
Since 
\begin{equation*}
[\bm{m}\bm{P} + \bm{n}\bm{J}, \bm{k}\bm{P}+\bm{\ell}\bm{J}] = (\bm{m} \wedge \bm{\ell} + \bm{n} \wedge \bm{k})\bm{P} + (\bm{n} \wedge \bm{\ell})\bm{J}
\end{equation*}
in order for (\ref{basis}) to generate a Lie subalgebra, we require: 
\begin{enumerate}
\item $\bm{n} \wedge \bm{\ell} \in \text{Span}\{\bm{n} , \bm{\ell} \}$: Since $\langle \bm{n} \wedge \bm{\ell}, \bm{n} \rangle = \langle \bm{n} \wedge \bm{\ell}, \bm{\ell} \rangle = 0$, this condition holds if and only if $\bm{n} \wedge \bm{\ell}=0$, which implies (both in the Euclidean and Lorentzian settings) $\bm{\ell} = C \bm{n}$ for some $C \in \mathbb{R}$.
\item $\bm{m} \wedge \bm{\ell} + \bm{n} \wedge \bm{k} \in \text{Span}\{\bm{m},\bm{k}\}$: This condition holds if and only if there exist $A,B \in \mathbb{R}$ such that 
\begin{equation*}
    \bm{m} \wedge \bm{\ell} + \bm{n} \wedge \bm{k} = \bm{m} \wedge (C \bm{n}) + \bm{n} \wedge \bm{k}= (C \bm{m}-\bm{k}) \wedge \bm{n}= A \bm{k} + B \bm{m}
\end{equation*}
\end{enumerate}

From above we know $\mathfrak{g}^{(1)} \equiv [\mathfrak{g},\mathfrak{g}]$ is generated by 
\begin{equation*}
    (A \bm{k} +  B\bm{m})  \bm{P}
\end{equation*}
Similarly, by direct computation, we get $\mathfrak{g}^{(2)} \equiv [[\mathfrak{g},\mathfrak{g}],\mathfrak{g}]$ is generated by 
\begin{equation*}
 ((A \bm{k} + B \bm{m} ) \wedge \bm{n}) \bm{P}.   
\end{equation*}
Indeed, by induction, we have that any term $\mathfrak{g}^{(n)}$ in the lower central series is one dimensional and generated by 
\begin{equation*}
    (\cdots(((A \bm{k}+ B \bm{m}) \underbrace{\wedge \bm{n}) \wedge \bm{n}) \wedge \cdots \wedge \bm{n}}_{(n-1)-\text{times}})\bm{P}
\end{equation*}
Hence, the condition 
\begin{equation*}
 (C \bm{m} - \bm{k}) \wedge \bm{n}=0   
\end{equation*}
is required in order to have a Nilpotent Lie subalgebra. Since the condition $C \bm{m}-\bm{k}=0$ provides a one dimensional Lie subalgebra, we conclude that any the two dimensional Nilpotent Lie subalgebras of $\mathfrak{iso}(3)$ or $\mathfrak{iso}(2,1)$ are generated by sets of the form 
\begin{equation}\label{isosub}
    \{\bm{m}\bm{P}+\bm{n}\bm{J}, (C \bm{m}-D \bm{n})\bm{P}+ C \bm{n}\bm{J} \ | \ C,D \in \mathbb{R} \}
\end{equation}
Finally, assume there exists $\bm{s}P + \bm{t}J$ in the normalizer of the Lie algebra generated by (\ref{isosub}) that does not belong to it. The brackets of this element with the generators of the Lie algebra are given by 
\begin{equation*}
  [\bm{m}\bm{P}+\bm{n}\bm{J}, \bm{s}\bm{P}+\bm{t}\bm{J}]=(\bm{m} \wedge \bm{t} + \bm{n} \wedge \bm{s})\bm{P} + (\bm{n} \wedge \bm{t})\bm{J}  
\end{equation*}
and 
\begin{equation*}
 [(C \bm{m}-D \bm{n})\bm{P}+C\bm{n}\bm{J}, \bm{s}\bm{P}+\bm{t}\bm{J}]=[(C \bm{m}-D \bm{n}) \wedge \bm{t} + (C \bm{n} \wedge \bm{s})]\bm{P} + (C\bm{n} \wedge \bm{t})\bm{J}   
\end{equation*}
The fact the right hand side belongs to the Lie algebra generated by (\ref{isosub}) implies
    \begin{equation*}
        \bm{t}=E \bm{n} \qquad \text{for } E \in \mathbb{R}
    \end{equation*}
Hence it follows that $\bm{s}\bm{P} + \bm{t}\bm{J}$ must be of the form $\bm{s}\bm{P} + E\bm{n}\bm{J}$, such that the previous two Lie brackets reduce then to 
\begin{equation*}
  [\bm{m}\bm{P}+\bm{n}\bm{J}, \bm{s}\bm{P}+E\bm{n}\bm{J}]=(E \bm{m} \wedge \bm{n} + \bm{n} \wedge \bm{s})\bm{P}   
\end{equation*}
and 
\begin{equation*}
 [(C \bm{m}-D \bm{n})\bm{P}+C\bm{n}\bm{J}, \bm{s}\bm{P}+E\bm{n}\bm{J}]=[E(C \bm{m}-D \bm{n}) \wedge \bm{n} + (C \bm{n} \wedge \bm{s})]\bm{P}    
\end{equation*}

Here we split in 2 cases:
\begin{enumerate}
    \item If $\bm{n}^2 = 0$, then we find that if we take
    \begin{equation*}
        \bm{s}=E \bm{m} - D \bm{n} + G \bm{m} \wedge \bm{n}
    \end{equation*}
    then 
    \begin{equation*}
  [\bm{m}\bm{P}+\bm{n}\bm{J}, \bm{s}\bm{P}+E\bm{n}\bm{J}]= - G\langle \bm{m},\bm{n} \rangle \bm{n}\bm{P}
\end{equation*}
and 
\begin{equation*}
 [(C \bm{m}-D \bm{n})\bm{P}+C\bm{n}\bm{J}, \bm{s}\bm{P}+E\bm{n}\bm{J}]= - C G \langle \bm{m},\bm{n} \rangle \bm{n}\bm{P}
\end{equation*}
concluding that in this case is possible to find an element in the normalizer that indeed does not belong to the Lie subalgebra generated by (\ref{isosub}), say 
\begin{equation*}
    (E \bm{m}- F \bm{n} + G \bm{m} \wedge \bm{n})\bm{P} + E \bm{n}\bm{J}
\end{equation*}
with $G \neq 0$.
\item If $\bm{n}^2 \neq 0$ and $C \neq 0$, then the only possibility is
\begin{equation*}
\bm{s}=E \bm{m}-F \bm{n}
\end{equation*}
and so an element in the normalizer must be of the form 
\begin{equation*}
 (E \bm{m}- F \bm{n})\bm{P} + E \bm{n}\bm{J}    
\end{equation*}
which clearly belongs to the span of (\ref{isosub}).
\end{enumerate}
\end{proof}
In the Euclidean framework, by conjugating with the right element in $\text{ISO}(3)$, any element in (\ref{cartaniso2}) could be mapped into the element with $\bm{n}=\bm{0}$ and  $\bm{m} \in \text{Span}\{\bm{e}_0\}$. Analogously for the Lorentzian signature, by conjugating with the proper element in $\text{ISO}(3)$ any element in (\ref{cartaniso2}) could be mapped either the element with $\bm{n}=\bm{0}$ and $\bm{m} \in \text{Span}\{\bm{e}_0\}$ (if $\bm{m}^2>0$) or the one with $\bm{n}=\bm{0}$ and  $\bm{m} \in \text{Span}\{ \bm{e}_1 \}$ (if $\bm{m}^2<0$). Hence any Cartan subalgebra of $\mathfrak{iso}(3)$ is conjugate-equivalent to $\text{Span}\{J_0,P_0\}$, while any Cartan subalgebra of $\mathfrak{iso}(2,1)$ is conjugate-equivalent to $\text{Span}\{J_0,P_0\}$ or $\text{Span}\{J_1,P_1\}$.

\section*{Appendix C: Weierstrass factorization theorem}\label{AppC}
In Section \ref{Sec4paper1} we made use of some expansions of certain (meromorphic) functions, in order to recognize that the dynamical generalized complexifications (which include the particular $\Lambda=0$ case found before in \cite{Meusburger:2012wc}) are dynamical gauge equivalent to the dynamical $r$-matrices studied by Feher, Gabor, Marshall, Palla and Pusztai in the setting of WZNW and Calogero-Moser models. \\\\
All the expansions are derived from a "\textit{well-known}" result in complex analysis known as the \textit{Weierstrass factorization Theorem} (see e.g. \cite{Entire}). It states the following: If $f: \mathbb{C} \to \mathbb{C}$ is an entire function with a zero at $z=0$ of order $m$ and with non-zero zeros $\{a_n\}$ (including multiplicities), then there exist an entire function $g:\mathbb{C} \to \mathbb{C}$ and a sequence of integers $\{p_k\}$, such that 
\begin{equation*}
    f(z)=z^m e^{g(z)} \prod_{k=1}^{\infty} E_{p_k}\left( \frac{z}{a_k} \right) 
\end{equation*}
where $E_n$ are the \textit{Weierstrass elementary factors}, given by
\begin{equation*}
E_n(z) = \begin{cases} (1-z) & \text{if } n=0 \\ (1-z) \exp \left( \sum\limits_{m=1}^n \frac{z^m}{m} \right) & \text{otherwise} \end{cases}.
\end{equation*}
For example (see \cite{Entire}), the trigonometric function $\sin(z)$ can be \textit{factorized} as
\begin{equation*}
\sin(z)=z \prod_{k=1}^\infty \left( 1-\frac{z^2}{k^2 \pi^2} \right).
\end{equation*}
and by by taking $\log$ and differentiating both sides 
we get 
\begin{equation}\label{Weierscot}
\cot(z)  - \frac{1}{z}  = 2 \sum_{n=1}^{\infty}\frac{z}{z^2-(\pi n)^2},
\end{equation}
which is used in the derivation of the functions $F(\xi,\zeta,\lambda)$ (\ref{eqFF:1}) and $G(\xi,\zeta,\lambda)$ (\ref{eqFF:2}) for the $\lambda \neq 0$ cases. Then, by replacing $z \mapsto iz$, we conclude
\begin{equation}\label{Weierscoth}
 \coth(z) - \frac{1}{z}  = 2 \sum_{n=1}^{\infty}\frac{z}{(\pi n)^2+z^2}
\end{equation}
which is precisely the representation of the function $g(z)$ used in (\ref{expancoth}). \\\\
Similarly, but using instead the Weierstrass expansion of the function $\cos(z)$ (see \cite{Entire}), we obtain 
\begin{equation}\label{Weierstan}
\tan \left(z \right) = -\frac{1}{z-\pi/2} + \sum_{n=1}^{\infty} \frac{2(z-\pi/2)}{( \pi n)^2-(z-\pi/2)^2}
\end{equation}
used in (\ref{eqf:1}) to find a compact form of $F(\xi,\zeta,0)$. \\\\
Finally, simply by taking the derivative of the previous expansion, we conclude 
\begin{equation}\label{Weierscos2}
\begin{split}
\frac{1}{ 2 \cos^2(z)}=  \left[ \frac{1}{(z-\pi/2)^2} + 2 \sum_{n=1}^{\infty}\frac{(z-\pi/2)^2+(\pi n)^2}{[(z-\pi/2)^2-(\pi n)^2]^2} \right] 
\end{split}
\end{equation}
used in (\ref{eqg:2}) for $G(\xi,\zeta,0)$. 

\section*{Appendix D: Functional coefficients for FGMPP classical dynamical \texorpdfstring{$(\mathfrak{g}_\lambda,\mathfrak{h}_{\lambda,0},K_{\alpha \beta})$}{} \texorpdfstring{$r$}{}-matrices}
In this brief appendix, we give some details on how the coefficients $F(\xi,\zeta,\lambda)$ and $G(\xi,\zeta,\lambda)$ in (\ref{FGMPPans}), for $\lambda \neq 0$, were obtained from (\ref{eqFF:1}) and (\ref{eqFF:2}). \\\\ 
In the case  $\lambda <0$ we express $\theta= \sqrt{|\lambda|}i$ such that $i^2=-1$ (formal symbol). By splitting the function $\cot(z)$ with $z \in \mathbb{R}_\lambda$ into its real and $i$- parts, we get 
\begin{equation*}
\begin{split}
\cot(\psi + \theta \gamma) & =\frac{\sin (\psi) \cos(\psi)}{\sin^2(\psi)+ \sinh^2(\sqrt{|\lambda|} \gamma)} - i \frac{ \sinh(\sqrt{|\lambda|} \gamma)\cosh(\sqrt{|\lambda|} \gamma)}{\sin^2( \psi)+ \sinh^2(\sqrt{|\lambda|} \gamma)} \\
& = \frac{\sin(2 \psi)}{\cosh(2 \sqrt{|\lambda|} \gamma)-\cos(2 \psi)} + i \frac{\sinh(2 \sqrt{|\lambda|} \gamma)}{\cos(2 \psi)-\cosh(2 \sqrt{|\lambda|} \gamma)},
\end{split}
\end{equation*}
and so, from (\ref{eqFF:1}), we obtain 
\begin{equation*}
\begin{split}
F(\xi,\zeta,\lambda) & = -\frac{1}{2} \left[ \cot \left( \xi+ \theta \zeta \right) + \cot \left( \xi - \theta \zeta \right) \right] \\
& = - \frac{\sin(2 \xi)}{-\cos(2 \xi)+ \cosh(2 \sqrt{|\lambda|} \zeta)},
\end{split}
\end{equation*}      
while from (\ref{eqFF:2}), 
\begin{equation*}
\begin{split}
G(\xi,\zeta,\lambda) & = -\frac{1}{2 \theta} \left[ \cot \left( \xi+ \theta \zeta \right) - \cot \left( \xi - \theta \zeta \right) \right] \\
& =  \frac{1}{2 \sqrt{|\lambda|}} \frac{\sinh(2\sqrt{|\lambda|} \zeta)}{-\cos(2\xi)+ \cosh(2\sqrt{|\lambda|}\zeta)}.
\end{split}
\end{equation*}      
Similarly, in the $\lambda>0$ case we express $\theta= \sqrt{\lambda} i$ such that $i^2=1$ (formal symbol). In this case, the splitting of the function $\cot(z)$ with $z \in \mathbb{R}_\lambda$ into its real and $i$- parts is given by
\begin{equation*}
\begin{split}
\cot(\psi + \theta \gamma) & =  \frac{\sin(2 \psi)}{\cos(2\sqrt{\lambda}\alpha)-\cos (2 \psi)} + i \frac{\sin(2 \sqrt{\lambda} \gamma)}{\cos(2 \psi)- \cos(2 \sqrt{\lambda} \gamma)} 
\end{split}
\end{equation*}
so in this case, from (\ref{eqFF:1}) we get 
\begin{equation*}
\begin{split}
F(\xi,\zeta,\lambda) & = - \frac{\sin(2 \xi)}{-\cos (2 \xi)+\cos(2 \sqrt{\lambda}\zeta)},
\end{split}
\end{equation*}
and from (\ref{eqFF:2}) 
\begin{equation*}
\begin{split}
G(\xi,\zeta,\lambda) & = \frac{1}{2
\sqrt{\lambda}} \frac{\sin(2 \sqrt{\lambda} \zeta)}{-\cos(2 \xi)+ \cos(2 \sqrt{\lambda}\zeta)}.
\end{split}
\end{equation*}


\printbibliography 

\end{document}